\documentclass[a4paper,10pt]{article}
\usepackage[utf8]{inputenc}
\usepackage{amsfonts}
\usepackage{amsthm}
\usepackage{amsmath}
\usepackage{amssymb}
\usepackage{setspace}
\usepackage{graphicx}
\usepackage{authblk}
\usepackage{array}
\usepackage{setspace}
\usepackage[margin=2.5cm]{geometry}
\usepackage[capitalise]{cleveref}
\usepackage[detect-all]{siunitx}
\usepackage{float}
\usepackage{mathtools}
\numberwithin{equation}{section}
\usepackage{longtable}
\usepackage{xcolor}
\newtheorem{theorem}{Theorem}
\newtheorem{lemma}[theorem]{Lemma}

\theoremstyle{definition}

\usepackage{appendix}
\theoremstyle{remark}

\numberwithin{equation}{section}



\title{Accurate transient heat flux from simple treatment of surface temperature distribution in the semi-infinite case}

\author[1]{David Buttsworth}
\author[2]{Timothy Buttsworth}
\affil[1]{School of Engineering, University of Southern Queensland}
\affil[2]{School of Mathematics and Statistics, The University of New South Wales}

\date{\today}

\begin{document}

\maketitle

\begin{abstract}
When the variations of surface temperature are measured both spatially and temporally, analytical expressions that correctly account for multi-dimensional transient conduction can be applied. 
To enhance the accessibility of these accurate multi-dimensional methods, expressions for converting between surface temperature and heat flux are presented as the sum of the one-dimensional component plus the multi-dimensional component.
Advantage arises herein because potential numerical challenges are isolated within the one-dimensional component and practitioners are already familiar with well-established one-dimensional methods.  
The second derivative of the surface heat flux distribution scaled by the thermal diffusivity and the duration of the experiment delivers an approximation of the multi-dimensional conduction term. 
For the analysis of experiments in which multi-dimensional effects are significant, a simplified numerical approach in which the temperature within each pixel is treated as uniform is demonstrated. 
The approach involves convolution of temperature differences and pixel-based impulse response functions, followed by a summation of results across the region of interest, but there are no singularities that require special treatment in the multi-dimensional component.
Recovery of heat flux distributions to within 1\,\% is demonstrated for two-dimensional heat flux distributions discretized using several tens of elements, and for a three-dimensional distribution discretized using several hundred pixels.
Higher accuracy can be achieved by using finer spatial resolution, but the level of discretization used herein is likely sufficient for practical applications since typical experimental uncertainties are much larger than 1\,\%. 
\end{abstract}


\section{Introduction}

Heat flux is frequently inferred from measurements of surface temperature history.
In the case of a semi-infinite surface with one-dimensional conduction,  
analysis methods include the Cook-Felderman \cite{Cook1966} approach, and the impulse response filtering method of Oldfield \cite{Oldfield2008}.

Multi-dimensional conduction effects can arise because the surface heat flux has a nonuniform distribution, or because the conducting solid has a finite size or is inhomogeneous. For a steady heat flux distribution, the magnitude of the error associated with multi-dimensional conduction effects increases with time, so provided the necessary temperature history data can be acquired in a short enough period of time, a one-dimensional approach may be sufficiently accurate. The validity of a one-dimensional method can initially be assessed by comparing relevant physical distances to the size of the heat kernel at the end of the experiment. For example, if experiments are performed using an insulating solid with a thermal diffusivity $\alpha = 1 \times 10^{-7}$m$^2$/s and a test period of 200\,ms, then the standard deviation of the heat kernel at this time is $\sqrt{2\alpha t}=0.2$\,mm. If the physical scale of variations in heat flux or the geometric features of the solid are about this size or smaller, a one-dimensional analysis is unlikely to be sufficient.

In cases where multi-dimensional conduction effects are important, the method introduced by Estorf \cite{estorf2006image} can be applied for the accurate treatment of surface thermal videography. The method accurately accounts for multi-dimensional conduction effects through application of a spatial integral with the appropriate Gaussian mask, followed by a convolution integral. These boundary integrals enable accurate deduction of flux while avoiding the computational cost of other approaches relying on simulations of the transient conduction throughout the solid.  A similar analysis was developed by Liu et al. \cite{liu2010analytical} for treating the more complicated case of a thin temperature-sensitive surface layer with different thermal properties from the substrate. An inverse methodology was advocated by Estorf \cite{estorf2006image} for improved treatment of noise in the case of homogeneous solids, and similar strategies have been explored in the case of two-layer work by Liu et al. \cite{liu2011correcting,liu2018analytical,liu2019inverse}.  The ITLR group at the University of Stuttgart have successfully applied the Estorf method to the treatment of multi-dimensional conduction effects that arise due to the relatively long duration (several minutes) of their experiments \cite{brack2022experimental,hartmann2024determination}.

Although the boundary integral methods introduced first by Estorf \cite{estorf2006image} and extended by Liu et al. \cite{liu2010analytical} appear well-suited to deduction of heat flux from surface thermal videography, the adoption of such methods has not been universal. For example, to accommodate the multi-dimensional conduction effects within a plate delivering a sonic jet into supersonic crossflow, Bae et al. \cite{bae2021measurement} defined the surface heat flux via a finite volume treatment of the plate (1.1 million cells were used) when Dirichlet boundary conditions based on the IR thermal videography were applied. 
Munoz et al. \cite{munoz2012instability} obtained IR videography on a 7~degree half-angle cone and followed the approach developed by Estorf \cite{estorf2006image} for temperature calibration, emissivity correction and spatial mapping, but rather than adopting the full-surface inverse method for deducing the flux, a direct one-dimensional finite difference analysis was used.
Zhang et al. \cite{zhang2021heat} similarly cite \cite{estorf2006image} as providing an option for treatment of multi-dimensional conduction effects in their cone model, but they instead proceed to apply a finite difference method that approximates the conduction as involving one-dimensional effects only.

While the Estorf method is correct and accurate, we believe that it
is not being widely adopted for several reasons. Firstly, there is the perception, if not the reality, that is is easier to discretize the necessary volume of the conducting solid and apply the measured temperature as the boundary condition and obtain the flux therefrom. Secondly, practitioners know thoroughly their own preferred, direct, one-dimensional heat conduction treatments, while multi-dimensional methods appear complex and require implementation of new techniques. Finally, errors associated with multi-dimensional conduction effects generally grow monotonically from commencement of the experiment, so errors due to multi-dimensional effects in long-duration experiments may cause insidious errors in the deduction of heat transfer coefficients and flow recovery temperatures. Some experimenters may even be tempted to avoid the hurdle of implementing the multi-dimensional analysis by performing short experiments only, failing to capitalize on opportunities to acquire accurate heat transfer coefficient and flow recovery temperature data.

Through this paper, we seek to encourage broader adoption of the multi-dimensional analysis by recasting the relevant multi-dimensional expressions in the form of the one-dimensional component plus the multi-dimensional component. The expressions are developed without invoking Fourier or Laplace transforms. Our formulas are presented in \cref{NtoD} and \cref{DtoN} below, with the derivation and the application of these formulas demonstrated in subsequent sections. Since research groups who perform transient heat transfer experiments will invariably have their own preferred methods for practical treatment of the one-dimensional conduction problem, the remaining task of implementing the multi-dimensional correction represents only modest effort. We believe the simplicity of the expression for the multi-dimensional conduction component will encourage broader adoption of these methods.

\section{Preliminaries, the formula and approach}\label{sec:prelim}
\subsection{Domain}
We use $\Omega=(\mathbb{R}^n)^+$ to denote \textit{upper-half space}, i.e., the set of points $(x_1,\cdots,x_n)\in \mathbb{R}^n$ with $x_n\ge 0$. The boundary of $(\mathbb{R}^n)^+$ is denoted $\partial \Omega$, and consists of the points $(x_1,\cdots,x_n)$ with $x_n=0$. We let $\nu=(0,0,\cdots,-1)$ be the outward-pointing unit normal vector to $\Omega$. Thus, we will use $\frac{\partial}{\partial \nu}$ or $\nabla_{\nu}$ to denote the partial differential operator $-\frac{\partial}{\partial x_n}$. In our discussion, we often consider functions of two spatial variables $x$ and $y$; therefore, we will sometimes use the notation $\frac{\partial}{\partial \nu_x}$ and $\nabla_{\nu_x}$ (or $\frac{\partial}{\partial \nu_y}$ and $\nabla_{\nu_y}$) to specify the variable in which the normal differentiation occurs. 
\subsection{Heat equation}
Choose an arbitrary maximal time $T_{\text{max}}>0$, and take a time-varying solution $u$ of the heat equation on $\Omega$, i.e., a smooth function $u:[0,T_{\text{max}}]\times \Omega\to \mathbb{R}$ which satisfies  
\begin{align}\label{heatequation}
\begin{split}
     \frac{\partial u}{\partial t}&=\Delta u \ \text{in} \ \Omega \ \text{(heat equation)},\\
    u(0,x)&=0 \ \text{for} \ x\in \Omega \ \text{(initial conditions)}.
    \end{split}
\end{align}
Two key quantities associated to the solution $u$ are the \textit{boundary temperature} $g$ and \textit{boundary flux} $p$. These are time-varying functions defined only for spatial points on the boundary (so we write $g,q:[0,T_{\text{max}}]\times \partial \Omega\to \mathbb{R}$), which are defined by
\begin{align}
    g(t,x)&=u(t,x), \qquad \text{(the temperature at the boundary), and} \label{DCs}\\
    q(t,x)&=\frac{\partial u(t,x)}{\partial \nu_x}, \qquad \text{(the flux at the boundary).} \label{NCs}
\end{align}

\subsection{Formulae}
The key formulae in this paper relate the surface temperature $g$ and the surface heat flux $q$:
\begin{align}\label{NtoD}
    g(t,x) & = \underbrace{\int_0^{t}  \frac{q(s,x)}{(\pi (t-s))^{1/2} } ds}_{\textrm{one-dimensional component, $g_{1d}$}} + \underbrace{\int_0^{t} \int_{\partial \Omega}   (q(s,y) - q(s,x)) \frac{\text{exp} \left(- \frac{|{x}-{y}|^2}{4(t-s)} \right)} {4(\pi (t-s))^\frac{n}{2}}  dy ds}_{\textrm{multi-dimensional component, $g_{md}$}} , 
\end{align}
and 
\begin{align}\label{DtoN}
\begin{split}
       q(t,x)
         &=\underbrace{\frac{g(t,x)}{\sqrt{\pi t}}+\int_0^t \frac{g(t,x)-g(s,x)}{\sqrt{4\pi} (t-s)^{3/2}}ds}_\textrm{one-dimensional component, $q_{1d}$} - \underbrace{\int_0^t \int_{\partial \Omega}(g(s,y)-g(s,x))\frac{\text{exp}\left(-\frac{\left|x-y\right|^2}{4(t-s)}\right)}{(t-s)(4\pi (t-s))^{\frac{n}{2}}}dyds}_\textrm{multi-dimensional component, $q_{md}$}.
         \end{split}
\end{align}
We have explicitly decomposed our formulae into two parts: the one-dimensional component and the multi-dimensional term.  In the case of the temperature result \cref{NtoD}, we see that the actual surface temperature rise at a given location $g(t,x)$ will be larger than indicated by the one-dimensional component if the heat flux elsewhere $q(s,y)$ is larger than at the location of interest $q(s,x)$. In the case of the heat flux result \cref{DtoN}, we see the actual heat flux at a given location $q(t,x)$ will be smaller than the one-dimensional component if the temperature rise elsewhere $g(s,y)$ is larger than at the location of interest $g(s,x)$.  

\subsection{Approach}

The one-dimensional (1d) plus multi-dimensional (md) decomposition leverages significant practical advantage for heat transfer analysis  because the singularity in the temperature-to-flux treatment is isolated within the one-dimensional component. 
Although the time-scaled Gaussian terms in the multi-dimensional components of \cref{NtoD} and \cref{DtoN} may appear challenging for $s = t$, the spatial integration proceeds first and delivers a result that is amenable to integration with time.
Transient heat flux experimenters are already well-equipped with the necessary tools for one-dimensional transient heat flux analyses, so the challenge in tackling the multi-dimensional problem is diminished.

In this paper, we derive \cref{NtoD} and \cref{DtoN}, and demonstrate that they are inverses of each other, and we introduce an elementary numerical treatment allowing flux to be recovered from temperature. We illustrate the accuracy that can be achieved through the elementary numerical treatment by: (1) assuming a certain distribution of flux $q$; (2) applying that flux $q$ in the flux-to-temperature result \cref{NtoD} to give $g$; (3) using $g$ in the temperature-to-flux result \cref{DtoN} to recover a heat flux $q_r$; and (4) identifying the magnitude of the error in the numerical treatment by comparing the recovered flux $q_r$ with assumed distribution $q$. 

\section{Derivation}\label{sec:derivation}
\subsection{Heat kernels}
The key ingredients in our derivation are the fundamental solutions of the Dirichlet and Neumann problems for the heat equations, which we will refer to as the \textit{Neumann} and \textit{Dirichlet heat kernels}, respectively: 
\begin{align*}
    k_{N}(t,x,y)&= k_{E}(t,x,y) + k_{E}(t,x,y^*), \\
    k_{D}(t,x,y)&= k_{E}(t,x,y) - k_{E}(t,x,y^*), 
\end{align*}
where the Euclidean heat kernel is 
\begin{align*}
    k_{E}(t,x,y)&= \frac{\text{exp}\left(\frac{-\left|x-y\right|^2}{4t}\right)}{(4\pi t)^{\frac{n}{2}}} 
\end{align*}
and the $*$ denotes reflection over the boundary $\partial\Omega$, i.e., $(y_1,\cdots,y_{n-1},y_n)^*=(y_1,\cdots,y_{n-1},-y_n)$. 
The Neumann and Dirichlet heat kernels have the following properties:
\begin{itemize}
    \item they both satisfy the heat equation in both spatial variables: $\frac{\partial k_N(t,x,y)}{\partial t}=\Delta_x k_N(t,x,y)=\Delta_y k_N(t,x,y)$, and similarly for $k_D$;
    \item for a fixed $y\notin \partial \Omega$, the initial conditions of $k_N(t,x,y)$ and $k_D(t,x,y)$ are unit energy compactified into the point $x=y$;
    \item both satisfy the relevant boundary conditions for $x\in \partial \Omega$ or $y\in \partial \Omega$. 
\end{itemize}
\subsection{Solving the Dirichlet and Neumann problems for temperature}
It is well-known that the solution of the Neumann problem (\eqref{heatequation} and \eqref{NCs}) is 
\begin{align}
\label{eq:NeumannConv}
    u(t,x)=\int_0^t \left(\int_{\partial \Omega}k_N(t-s,x,y)q(s,y)dy\right)ds.
\end{align}
In the above expression, we recognize the temperature effects arising from unit energy delivered at points on the surface $k_N$ are being scaled by the local surface distribution of heat flux $q(t,x)$ and integrated across the surface; the net temperature result at any point $u(t,x)$ is the convolution of the surface-integrated effects through earlier times.  
The variable $x$ can be taken all the way to the boundary in this formula, giving the formula \eqref{NtoD} for $g$ in terms of $q$. 

Perhaps surprisingly, finding an expression for $q$ in terms of $g$ is much more difficult than the other way around. 
We can similarly use the Dirichlet heat kernel to find a formula for $u$ in terms of $g$ (i.e., solving \eqref{heatequation} subject to the Dirichlet condition \eqref{DCs}):
\begin{align}\label{usolution}
    u(t,x)=\int_{0}^{t} \left( \int_{\partial \Omega}\frac{\partial k_D}{\partial \nu_y}(t-s,x,y)g(s,y)dy \right) ds.
\end{align}
Formula \eqref{usolution} illustrates what we expect; it is possible to uniquely determine the temperature at \textit{any} point $x$ using only data of the temperature at the boundary at previous times. 
However, the problem with \eqref{usolution} is that it only holds for $x$ in the interior, and breaks down at the boundary. Essentially, the new problem is that, while both the Neumann and Dirichlet heat kernels have singularities at $x=y$, $t=0$, the singularity of $\frac{\partial k_D}{\partial \nu_y}$ is not spatially integrable if $x$ is on the boundary, and more care is required.

\subsection{Solving the Dirichlet problem for flux}
To work around the inability to integrate the Dirichlet derivative on the boundary, we introduce an auxilary function $\tilde{u}:[0,T_{\text{max}}]\times \Omega\to \mathbb{R}$ so that $\tilde{u}(0,x)=0$ and $\tilde{u}(t,x)=g(t,x)$ for $x\in \partial \Omega$. So $\tilde{u}$ satisfies the same initial and boundary conditions as $u$, but does not necessarily satisfy the heat equation. Then the function $v=\tilde{u}-u$ satisfies 
\begin{align}\label{Duhameld}
    \left(\frac{\partial}{\partial t}-\Delta \right)v&= \left(\frac{\partial}{\partial t}-\Delta \right)\tilde{u}\\
     v(0,x)&=0 \ \text{for} \ x\in \Omega \ \text{(initial conditions)}\\
    v(t,x)&=0 \ \text{for} \ x\in \partial \Omega \ \text{(boundary conditions)}. 
\end{align}
\textit{Duhamel's principle} states that
\begin{align}\label{vsoln}
    v(t,x)=\int_0^t \int_{\Omega}k_D(t-s,x,y)\left(\frac{\partial}{\partial s}-\Delta_y \right)\tilde{u}(s,y)dyds. 
\end{align}
The good news about formula \eqref{vsoln} is that it works for $x$ all the way up to the boundary. And in fact, we are allowed to put the derivative inside the integral for this particular formula (see \cref{justify} for the justification):
\begin{align}\label{vdinside}
    \frac{\partial v}{\partial \nu_x}(t,x)=\int_0^t \int_{\Omega}p(t-s,x,y)\left(\frac{\partial}{\partial s}-\Delta \right)\tilde{u}(s,y)dyds, 
\end{align}
where $p(t,x,y)=\frac{\partial k_D(t,x,y)}{\partial \nu_x}$, and $x$ is assumed to be in the boundary $\partial \Omega$. Observe that 
\begin{align}\label{pprop}
\frac{\partial p}{\partial t}-\Delta_y p&=0\\
p(t,x,y)&=0 \ \text{if} \ y\in \partial \Omega. 
\end{align}
The first equation holds because $k_D(t,x,y)$ satisfies the heat equation in $t,y$ for each $x$ (so the same is true if we differentiate in $x$). The second condition is because $k_D$ satisfies Dirichlet conditions in $y$, so the same is true for the derivative in $x$. We now proceed to evaluate \eqref{vdinside}. Ignoring the $ds$ integral for the moment, and fixing $s<t$, we are away from the $p$ singularity, so 
\begin{align*}
  &\int_{\Omega}p(t-s,x,y)\left(\frac{\partial}{\partial s}-\Delta \right)\tilde{u}(s,y)dy
  \\&= \int_{\Omega}p(t-s,x,y)\frac{\partial \tilde{u}(s,y)}{\partial s} dy- \int_{\Omega}p(t-s,x,y)\Delta\tilde{u}(s,y)dy\\
  &=\frac{\partial}{\partial s}\left(\int_{\Omega}p(t-s,x,y)\tilde{u}(s,y)dy\right)+\int_{\Omega}\frac{\partial }{\partial t}p(t-s,x,y)\tilde{u}(s,y)dy- \int_{\Omega}p(t-s,x,y)\Delta\tilde{u}(s,y)dy\\
  &=\frac{\partial}{\partial s}\left(\int_{\Omega}p(t-s,x,y)\tilde{u}(s,y)dy\right)+\int_{\Omega}\left(\frac{\partial }{\partial t}-\Delta_y\right)p(t-s,x,y)\tilde{u}(s,y)dy+\int_{\partial \Omega}g(s,y)\nabla_{\nu_y}p(t-s,x,y)dy\\
  &=\frac{\partial}{\partial s}\left(\int_{\Omega}p(t-s,x,y)\tilde{u}(s,y)dy\right)+\int_{\partial \Omega}g(s,y)\nabla_{\nu_y}p(t-s,x,y)dy,
\end{align*}
where in the second last line, we used Green's second integral identity, coupled with the fact that $p(t-s,x,y)=0$ if $y\in \partial \Omega$, and in the last line, we used the fact that $p$ satisfies the heat equation in $t,y$. 
We can finally re-introduce the $ds$ integral, using $\tilde{u}(0,y)=0$ to find 
\begin{align}\label{fluxsplit}
    \frac{\partial v(t,x)}{\partial \nu_x}&=\lim_{\tau\to t}\left(\underbrace{\int_{\Omega}p(t-\tau,x,y)\tilde{u}(\tau,y)dy}_{I_1}+\underbrace{\int_0^{\tau}\int_{\partial \Omega}g(s,y)\nabla_{\nu_y}p(t-s,x,y)dyds}_{I_2}\right).
\end{align}
It would be nice to take this limit by simply replacing all of the $\tau$s with $t$s. However, this must fail because both $\lim_{\tau\to t}I_1$ and $\lim_{\tau\to t}I_2$ diverge; it is only their \textit{sum} which converges. Thus, to proceed, we will decompose $I_1$ and $I_2$ into more convenient converging and diverging terms:
\begin{align*}
    I_1&=\int_{\Omega}p(t-\tau,x,y)\tilde{u}(\tau,y)dy\\
    &=\int_{\Omega}p(t-\tau,x,y)(\tilde{u}(\tau,y)-g(\tau,x))dy+(g(\tau,x)-g(t,x))\int_{\Omega}p(t-\tau,x,y)dy+g(t,x)\int_{\Omega}p(t-\tau,x,y)dy 
\end{align*}
and
\begin{align*}
    I_2&=\int_{0}^{\tau}\int_{\partial \Omega}(g(s,y)-g(s,x))\nabla_{\nu_y}p(t-s,x,y)dyds\\
    &+ \int_{0}^{\tau}(g(s,x)-g(t,x))\int_{\partial \Omega}\nabla_{\nu_y}p(t-s,x,y)dyds\\
    &+ g(t,x)\int_{0}^{\tau}\int_{\partial \Omega}\nabla_{\nu_y}p(t-s,x,y)dyds.
\end{align*}
We have $\lim_{\tau\to t}(g(\tau,x)-g(t,x))\int_{\Omega}p(t-\tau,x,y)dy=0$, assuming $g$ is differentiable (see \cref{justify} for some more details). 
Thus, \eqref{fluxsplit} 
becomes
 \begin{align*}
    -\frac{\partial u(t,x)}{\partial \nu_x}&=\lim_{\tau\to t}\left(\int_{\Omega}p(t-\tau,x,y)(\tilde{u}(\tau,y)-g(\tau,x))dy-\frac{\partial \tilde{u}}{\partial \nu_x}(t,x)\right)\\
    &+g(t,x)\lim_{\tau \to t}\left(\int_{\Omega}p(t-\tau,x,y)dy+\int_{0}^{\tau}\int_{\partial \Omega}\nabla_{\nu_y}p(t-s,x,y)dyds\right)\\
    &+ \int_{0}^{t}(g(s,x)-g(t,x))\int_{\partial \Omega}\nabla_{\nu_y}p(t-s,x,y)dyds\\
    &+\int_{0}^{t}\int_{\partial \Omega}(g(s,y)-g(s,x))\nabla_{\nu_y}p(t-s,x,y)dyds.
\end{align*}
It turns out that the first limit vanishes (again, see \cref{justify} for more details), so we finally arrive at 
\begin{align}
\begin{split}
    &q(t,x)\\
    &=-\underbrace{g(t,x)\lim_{\tau \to t}\left(\int_{0}^{\tau}\int_{\partial \Omega}\nabla_{\nu_y}p(t-s,x,y)dyds+\int_{\Omega}p(t-\tau,x,y)dy\right)}_{\text{constant contribution}}\\
    &+ \underbrace{\int_{0}^{t}(g(t,x)-g(s,x))\int_{\partial \Omega}\nabla_{\nu_y}p(t-s,x,y)dyds}_{\text{contribution from temporal variance at point $x$}}\\
    &+\underbrace{\int_{0}^{t}\int_{\partial \Omega}(g(s,x)-g(s,y))\nabla_{\nu_y}p(t-s,x,y)dyds}_{\text{multi-dimensional contribution}}.
    \end{split}
\end{align}
Let us compute these terms. First, when $x_n=0$ (so $x\in \partial \Omega$) 
\begin{align*}
    p(t,x,y)=\frac{\partial k_D}{\partial \nu_x}(t,x,y)=-\frac{y_n}{t(4\pi t)^{\frac{n}{2}}}\text{exp}\left(\frac{-\left|x-y\right|^2}{4t}\right),
\end{align*}
since $\left|x-y^*\right|=\left|x-y\right|$ if one of $x,y$ is on the boundary. 
Therefore, for $y$ on the boundary ($y_n=0$), we also have 
\begin{align*}
    \frac{\partial p}{\partial \nu_y}(t,x,y)=\frac{1}{t(4\pi t)^{\frac{n}{2}}}\text{exp}\left(\frac{-\left|x-y\right|^2}{4t}\right).
\end{align*}
We compute 
\begin{align*}
    \int_{\partial \Omega}\frac{\partial p}{\partial \nu_y}(t,x,y)dy=\frac{1}{t\sqrt{4\pi t}},
\end{align*}
so 
\begin{align*}
    \int_0^{\tau}\int_{\partial \Omega} \nabla_{\nu_y}p(t-s,x,y)dy=\frac{1}{\sqrt{\pi(t-\tau)}}-\frac{1}{\sqrt{\pi t}},
\end{align*}
and also
\begin{align*}
    \int_{\Omega}p(t-\tau,x,y)dy&=-\frac{1}{(t-\tau)\sqrt{4\pi (t-\tau)}}\int_{0}^{\infty}y_n \text{exp}\left(\frac{-y_n^2}{4(t-\tau)}\right)dy_n\\
    &=-\frac{1}{\sqrt{4\pi (t-\tau)}}\int_0^{\infty}z e^{-z^2/4}dz=-\frac{1}{\sqrt{\pi (t-\tau)}}.
\end{align*}

Therefore our formula is 
\begin{align*}
       q(t,x)
         &=\underbrace{\frac{g(t,x)}{\sqrt{\pi t}}+\int_0^t \frac{g(t,x)-g(s,x)}{\sqrt{4\pi} (t-s)^{3/2}}ds}_\textrm{one-dimensional component, $q_{1d}$}-\underbrace{\int_0^t \int_{\partial \Omega}(g(s,y)-g(s,x))\frac{\text{exp}\left(-\frac{\left|x-y\right|^2}{4(t-s)}\right)}{(t-s)(4\pi (t-s))^{\frac{n}{2}}}dyds}_\textrm{multi-dimensional component, $q_{md}$}.
\end{align*}

The above expression indicates the actual flux perpendicular to the surface can be obtained by accounting for the multi-dimensional component, treating it as a correction to the one-dimensional component. The form of the above expression which separates the one-dimensional and multi-dimensional components has merit because: (1) practitioners in the field of transient heat transfer measurements will be familiar with the one-dimensional component; (2) any preferred approach can be used for obtaining the apparent flux from temperature measurements by treating the conduction as though it were one-dimensional; and (3) there is no singularity in the multi-dimensional component.  Therefore, the only challenge for practitioners seeking high accuracy heat flux results throughout their relatively long-duration experiments becomes that of demonstrating the successful treatment of the multi-dimensional component only. 

\section{Illustration for continuous distributions}

Transient heat conduction experiments provide surface temperature measurements at discrete locations only, and while sub-millimeter spatial resolution can be achieved sometimes (for example, in the case of IR videography, depending on the optical magnification), the data only approximate the actual continuous distributions. Before proceeding to the treatment of discrete temperature measurements, examples of continuous distributions are considered -- a linear flux and a quadratic flux -- to demonstrate the application of the results, \cref{NtoD} and \cref{DtoN}.

\subsection{Linear flux}
\label{sec:linearflux}

Let us suppose we have a time-independent flux at the boundary that is given by 
\begin{align*}
    q(t,x) = v\cdot x,
\end{align*}
where $v=(v_1,\cdots,v_{n-1},0)$ is some fixed vector which is tangential to the boundary. The Neumann heat kernel, which for $y\in \partial \Omega$ can be written  
\begin{align*}
    k_{N}(t,x,y)&=2\frac{\text{exp}\left(\frac{-\left|x-y\right|^2}{4t}\right)}{(4\pi t)^{\frac{n}{2}}},
\end{align*}
is applied in \cref{eq:NeumannConv} to give the surface temperature as 
\begin{align*}
    g(t,x)&=\int_0^t \int_{\partial \Omega} 2v\cdot y \frac{\text{exp}\left(\frac{-\left|x-y\right|^2}{4(t-s)}\right)}{(4\pi (t-s))^{\frac{n}{2}}} dyds,
\end{align*}
and making a change of variable $z = (y-x)/\sqrt{t-s}$,
\begin{align*}
     g(t,x) &=\frac{2}{(4\pi)^{n/2}}\int_0^t\frac{1}{\sqrt{t-s}} \int_{\partial \Omega} v\cdot (z\sqrt{t-s}+x)e^{-\left|z\right|^2/4}dzds\\
    &=\frac{2v\cdot x (4\pi )^{n/2-1/2}}{(4\pi)^{n/2}}\int_0^t \frac{1}{\sqrt{t-s}}ds\\
    &=\frac{2v\cdot x \sqrt{t}}{ \sqrt{\pi}}.
\end{align*}
We see this result is the same as that obtained from the one-dimensional component of \cref{NtoD}, confirming expectations that there is no multi-dimensional contribution in the case of a linear distribution of flux. Turning attention to \cref{DtoN} to hopefully recover the correct time-independent flux on the boundary, we can see that the multi-dimensional component of flux vanishes because 
\begin{align*}
q_{md} &= \int_0^t \int_{\partial \Omega}(g(s,y)-g(s,x))\frac{\text{exp}\left(-\frac{\left|x-y\right|^2}{4(t-s)}\right)}{(t-s)(4\pi (t-s))^{\frac{n}{2}}}dy ds \\
      &=  \int_0^t \frac{2\sqrt{s}}{\sqrt{\pi}} \underbrace{ \int_{\partial \Omega} v \cdot \left( y  - x \right) \frac{\text{exp}\left(-\frac{\left|x-y\right|^2}{4(t-s)}\right)}{(t-s)(4\pi (t-s))^{\frac{n}{2}}}dy }_{=0} ds\\
      &=0.
\end{align*}
Therefore, in this case, \cref{DtoN} reduces to the one-dimensional component only and this gives
\begin{align*}
    q(t,x) &= \frac{g(t,x)}{\sqrt{\pi t}}+\int_0^t \frac{g(t,x)-g(s,x)}{\sqrt{4\pi} (t-s)^{3/2}}ds \\
    &= \frac{2v\cdot x}{\pi}-\frac{v\cdot x}{\pi}\int_0^t \frac{\sqrt{s}-\sqrt{t}}{(t-s)^{3/2}}ds\\
    &=v\cdot x \left(\frac{2+2\arcsin(1)-2}{\pi}\right)\\
    &=v\cdot x ,
\end{align*}
as required. 

\subsection{Quadratic flux}
\label{sec:quadraticflux}
To demonstrate the multi-dimensional components of the analysis, we need to consider a non-linear distribution of flux. Let us work in two-dimensions, and suppose the boundary flux is given by 
\begin{align*}
    q(t,x_1) = x_1^2.
\end{align*}
Applying \cref{eq:NeumannConv} gives the boundary temperature
\begin{align*}
    g(t,x_1)&=\int_0^t \int_{-\infty}^{\infty} 2 y_1^2 \frac{\text{exp}\left(\frac{-(x_1-y_1)^2}{4(t-s)}\right)}{4\pi (t-s)} dy_1 ds,
\end{align*}
and making a change of variable $z = (y_1-x_1)/\sqrt{t-s}$ enables development of an explicit expression,
\begin{align*}
    g(t,x_1)
    &=\frac{1}{2\pi}\int_0^t \frac{1}{\sqrt{t-s}}\int_{-\infty}^{\infty}(z\sqrt{t-s}+x_1)^2e^{-z^2/4}dzds\\
    &=\frac{1}{2\pi}\int_0^t \frac{1}{\sqrt{t-s}}\int_{-\infty}^{\infty}e^{-z^2/4}\left(x_1^2+z^2(t-s)\right)dzds\\
    &=\frac{1}{2\pi}\int_0^t \frac{1}{\sqrt{t-s}}\left(2x_1^2\sqrt{\pi}+4(t-s)\sqrt{\pi}\right)ds\\
    &= \frac{2}{\sqrt{\pi}}\left(x_1^2\sqrt{t}+\frac{2}{3}t^{3/2}\right)\\
    &=\frac{2}{\sqrt{\pi}} x_1^2 \sqrt{t} + \frac{4}{3 \sqrt{\pi}} t^{3/2}.
\end{align*}
We note the first and second terms in the above expression are the one-dimensional and multi-dimensional contributions corresponding to the first and second terms in \cref{NtoD}. In this quadratic heat flux case, the temperature rise at any location is larger than would be inferred from the one-dimensional analysis alone.
    
Now we can work backwards using the two-dimensional version of \cref{DtoN}, 
\begin{align*}
q(t,x_1) = \underbrace{\frac{g(t,x_1)}{\sqrt{\pi t}}+\int_0^t \frac{g(t,x_1)-g(s,x_1)}{\sqrt{4\pi} (t-s)^{3/2}}ds}_\textrm{one-dimensional component, $q_{1d}$}-\underbrace{\int_0^t \int_{-\infty}^{\infty}(g(s,y_1)-g(s,x_1))\frac{\text{exp}\left(-\frac{(x_1-y_1)^2}{4(t-s)}\right)}{4\pi (t-s)^2}dy_1ds}_\textrm{multi-dimensional component, $q_{md}$} ,
\end{align*}
to see if we can correctly recover the assumed quadratic flux.
Considering the one-dimensional component first, and substituting the calculated temperature rise $g(t,x_1)$ we have
\begin{align*}
    q_{1d} & = \frac{2}{\pi}\left(x_1^2+\frac{2t}{3}\right)+\int_0^t \frac{x_1^2(\sqrt{t}-\sqrt{s})+\frac{2(t^{3/2}-s^{3/2})}{3}}{\pi (t-s)^{3/2}}ds \\
    &= \frac{2}{\pi}\left(x_1^2+\frac{2t}{3}\right)+\frac{1}{\pi}\left(x_1^2(\pi-2)+\frac{2}{3}\left(\frac{3\pi t}{2}-2t\right)\right) \\
    &= x_1^2 + t
\end{align*}
The multi-dimensional component is 
\begin{align*}
    q_{md}&=\int_0^t \int_{-\infty}^\infty \frac{2\sqrt{s}}{\sqrt{\pi}}(y_1^2-x_1^2) \frac{\text{exp}\left(-\frac{(x_1-y_1)^2}{4(t-s)}\right)}{4\pi (t-s)^2} dy_1 ds, \\
    \end{align*}
    and making a change of variable $z = (y_1-x_1)/\sqrt{t-s}$,
    \begin{align*}
    q_{md} &= \int_0^t \frac{2\sqrt{s}}{\sqrt{\pi}} \int_{-\infty}^{\infty}(2 x_1 z + \sqrt{t-s} z^2) \frac{e^{-z^2/4}}{4\pi (t-s)} dzds \\
    &=\int_0^t\frac{2\sqrt{s}}{\sqrt{\pi(t-s)}}\int_{-\infty}^{\infty} \left( \frac{2x_1}{\sqrt{t-s}} z +  z^2 \right) \frac{e^{-z^2/4}}{4\pi} dzds\\
        &=\int_0^t\frac{2\sqrt{s}}{\pi\sqrt{(t-s)}}ds\\
        &=t.
\end{align*}
Thus, the net result is
\begin{align*}
    q(t,x_1) &= \underbrace{x_1^2 + t}_{\textrm{1d}} - \underbrace{t}_{\textrm{md}} = x_1^2,
\end{align*}
as required.

\subsection{Application: assessment of experiments}

In the case of the linear distribution of flux (\cref{sec:linearflux}), there will be lateral conduction effects since there is a spatial variation of surface temperature. However, these lateral conduction effects do not manifest as a contribution to the multi-dimensional conduction component. To assess the magnitude of multi-dimensional conduction effects in experiments, we need to model the surface heat flux using a nonlinear distribution, and the quadratic results (\cref{sec:quadraticflux}) provide a convenient approach.

From the above quadratic results, we see that if the scale-invariant flux is $q = C x_1^2$, where $C$ is some constant, then the multi-dimensional term will be $q_{md} = C t$.
We can use this result to estimate the significance of multi-dimensional effects in actual experiments if we equate the magnitude of the constant $C$ to the second derivative of surface heat flux in those experiments,
\begin{align*}
    C = \frac{d^2q}{dx_1^2} .
\end{align*}
We convert between scale-invariant quantities and scale-specific quantities, recognizing the scale-invariant time $t$ is related to scale-specific time $t_s$ via
\begin{align*}
    t = \frac{\alpha t_s}{L^2},
\end{align*}
where $\alpha$ is the thermal diffusivity (units, m$^2$/s) and $L$ is some suitable reference length (units, m) which provides the scale-invariant distance as 
\begin{align*}
    x = \frac{x_s}{L},
\end{align*}
where $x_s$ is the scale-specific distance (units, m). The scale-invariant heat flux $q$ is related to scale-specific quantities 
\begin{align*}
    q = \frac{q_s L}{k T_{i}} ,
\end{align*}
where $q_s$ is the scale-specific heat flux (units, W/m$^2$), $k$ is the thermal conductivity (units, W/m-K), and $T_i$ is the initial temperature of the solid. 

Based on the quadratic result in two dimensions (\cref{sec:quadraticflux}), we can estimate the magnitude of the multi-dimensional heat conduction effects from the local second derivative of the surface heat flux distribution, the thermal diffusivity, and the duration of the experiment using
\begin{align*}
    q_{s_{md}} &= \frac{d^2q_s}{d x_s^2} \alpha t_s ,
\end{align*}
and the relative error in the one-dimensional approach given by
\begin{align*}
    \text{(error)}_{1d} &= \frac{1}{q_s} \frac{d^2 q_s}{d x_s^2} \alpha t_s . 
\end{align*}

Taking the experiments by Zhang et al. \cite{zhang2021heat} as a specific illustration, we estimate that the lateral features in their figure 4 (wavy wall case) have a wavelength of about 5\,mm with peaks and troughs in heat flux being around 4.5 and 3.0 kW/m$^2$, respectively. For these conditions we estimate the largest magnitude of second derivative of the heat flux to be about $1.2\times10^6$\,kW/m$^4$. The model in their experiments is made from PEEK, for which we take the thermal diffusivity to be $\alpha = 1.9 \times 10^{-7}$\,m$^2$/s, and the duration of the experiment is 12\,s, so we estimate $ q_{s_{md}} = 2.8$\,kW/m$^2$. The estimated average relative error in the one-dimensional analysis due to multi-dimensional conduction effects is therefore around $(2.8)/(3.75)= 0.75$, so the one-dimensional approach clearly becomes unreliable by the end of their experiment. Accommodating the multi-dimensional conduction effects in future work will greatly improve the accuracy of their results.

\section{Illustration for discretised distributions}
\label{sec:discretised}

\subsection{Pixel-based data}

In transient heat transfer experiments, surface temperature measurements can be acquired using optical techniques including thermochromic liquid crystals, temperature sensitive paint, and thermal IR methods. In all cases, the surface signals will be sampled discretely at picture elements, \emph{pixels}. The counts (digital units) acquired at each pixel will be related to a temperature value via a suitable calibration, and each pixel will represent a certain area of the surface, depending on the optical projection. Arrays of values that are acquired within a relatively short period of time generate a \emph{frame} within a video sequence that is typically sampled at a constant rate (frames/second) for any given experiment.

In the present work, we consider the mapping from the physical surface to the imaging sensor to be constant: each pixels represents a specific area of the surface, and this area is constant for all pixels and is square. In practice, optical distortion effects due to surface curvature and lens effects will be present, but can be accommodated through calibration. The region of pixels mapped onto a subset of the semi-infinite surface that we consider is represented in \cref{fig:pixel_interpretation}. The edge length of each pixel in the surface-projections is $\Delta$. The location of interest is $x = (x_1,x_2)$, and the task is to determine the effects that all other pixels, of which $y_p = (y_{1p},y_{2p})$ is representative, will have on the results at $x$.

\begin{figure}[htp]
    \centering
    \includegraphics[width=79mm]{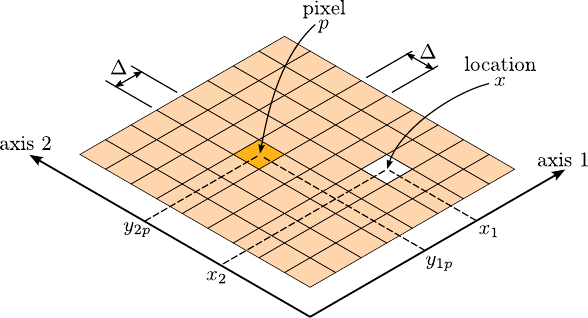} 
    \caption{Illustration of the surface on which heat transfer occurs discretized using pixels.}
    \label{fig:pixel_interpretation}
\end{figure}

\subsection{Transforming from flux to temperature}
\label{sec:flux-to-temperature}

For three dimensional heat conduction, the flux-to-temperature results \cref{NtoD} for the location of interest $x$ can be written
\begin{align*}
   g(t,x) & = \underbrace{\int_0^{t}  \frac{q(s,x)}{(\pi (t-s))^{1/2} } ds}_{\textrm{one-dimensional component, $g_{1d}$}} + \underbrace{\int_0^{t} \int_{-\infty}^{\infty} \int_{-\infty}^{\infty}   (q(s,y) - q(s,x)) \frac{\text{exp} \left(- \frac{|{x}-{y}|^2}{4(t-s)} \right)} {4(\pi (t-s))^3/2}  dy_1 dy_2 ds}_{\textrm{multi-dimensional component, $g_{md}$}} . \\  
\end{align*}
Since we assume the one-dimensional component is readily analysed with existing methods, the multi-dimensional component of the surface temperature is of primary interest and can be written
\begin{align*}
    g_{md}(t,x) &= \int_0^{t} \int_{-\infty}^{\infty} \int_{-\infty}^{\infty}   (q(t,y) - q(t,x)) \frac{\textrm{e}^{-((x_1-y_1)^2+(x_2-y_2)^2)/4(t-s)}}{4 (\pi (t-s))^{3/2}}  dy_1 dy_2 ds .
\end{align*}
We are considering square pixels of edge length $\Delta$, and we approximate the flux as spatially uniform over each pixel, so the multi-dimensional component of the temperature change for the location of interest can be approximated as  
\begin{align*}
    g_{md}(t,x)  & = \sum_{p = 1}^{n_p} \int_0^t  (q(s,y_p)-q(s,x)) \int_{y_{2p}-\frac{\Delta}{2} }^{y_{2p}+\frac{\Delta}{2}} \int_{y_{1p}-\frac{\Delta}{2}}^{y_{1p}+\frac{\Delta}{2}} \frac{e^{-((x_1-y_1)^2+(x_2-y_2)^2)/4(t-s)}}{4(\pi(t-s))^{3/2}} dy_1 dy_2 ds .
\end{align*}
where the summation is taken over the required pixels of which there is a total number $n_p$. It is convenient to define the following integral, 
\begin{align}
\label{eq:Ipixel}
I(t,z,z_p,\delta) & := \int_{z_{p}-\frac{\delta}{2} }^{z_{p}+\frac{\delta}{2}} \frac{e^{-(z-z_p)^2/4t}}{\sqrt{\pi t}} dz_p = \textrm{erf}\left( \frac{z-z_p+\frac{\delta}{2}}{\sqrt{4t}} \right) - \textrm{erf}\left( \frac{z-z_p-\frac{\delta}{2}}{\sqrt{4t}} \right) ,
\end{align}
so we have the multi-dimensional temperature effect as
\begin{align}
\label{eq:g_md}
    g_{md}(t,x)& =  \sum_{p = 1}^{n_p} \int_0^t  \underbrace{\left( q(s,y_p)-q(s,x) \right)}_{\textrm{forcing function}} \underbrace{\frac{I(t-s,x_1,y_{1p},\Delta) I(t-s,x_2,y_{2p},\Delta)}{2\sqrt{\pi} (4(t-s))^{1/2}}}_{\textrm{pixel-based impulse response function, $H_g$}} ds .
\end{align}
This expression indicates that the multi-dimensional temperature effects at any location of interest $x$ are obtained by convolving the relative heat flux at other pixels $q(s,y_p)-q(s,x)$ with the pixel-based temperature impulse response function $H_g$ and performing a summation across pixels of interest. 
The above expression was developed for the case of three-dimensional conduction, so in this case we write the three-dimensional version of the impulse response function,
\begin{align}
\label{eq:Hg3}
    H_{g3}(t,x,y_p,\Delta) &= \frac{1}{2 \sqrt{\pi}} \frac{I(t,x_1,y_{1p},\Delta) I(t,x_2,y_{2p},\Delta)}{(4t)^{1/2}},
\end{align}
and note that in the two-dimensional case, the relevant impulse response function is,
\begin{align}
\label{eq:Hg2}
    H_{g2}(t,x,y_p,\Delta) &= \frac{1}{\sqrt{\pi}} \frac{I(t,x,y_{p},\Delta)}{(4t)^{1/2}}.
\end{align}

Note that $H_{g3}$ and $H_{g2}$ have finite values at all times, provided $y_p \ne x$. The summation shown in \cref{eq:g_md} is taken across all pixels, but it does not actually need to include the case $y_p=x$ because for this location, the forcing function $q(s,y_p)-q(s,x)=0$, so there is zero contribution here. Thus, the infinite value of the impulse response functions for $y_p=x$ at $t=0$ is of no consequence.
A selection of the three-dimensional impulse response functions for converting from heat flux to temperature $H_{g3}$ are illustrated in \cref{fig:impulse_response_functions}(a) for the three pixels closest to $x = 0$. In all cases, the magnitude of the response is zero at $t=0$, consistent with the expectation that heat flux delivered at pixels $y_p \ne x$, will not an immediate effect on the temperature rise as $x$. The impulse response due to flux at the nearest pixel reaches a maximum value at the pixel-scaled time of $4t/\Delta^2 = 0.4304$ and has a value of $\Delta H_{g3} = 0.1730$. The corresponding results for the two-dimensional case are $4t/\Delta^2 = 0.7054$ with a maximum value of $\Delta H_{g2} = 0.2608$.

\begin{figure}[H]
    \centering
    \includegraphics[width=79mm]{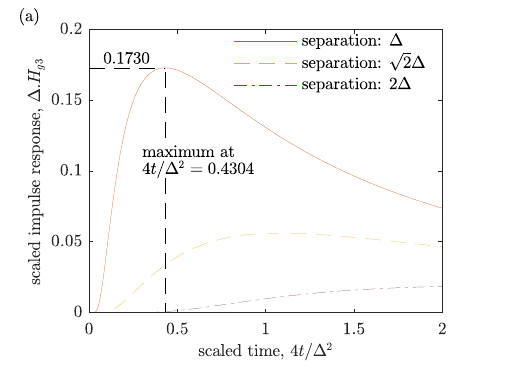}
    \includegraphics[width=79mm]{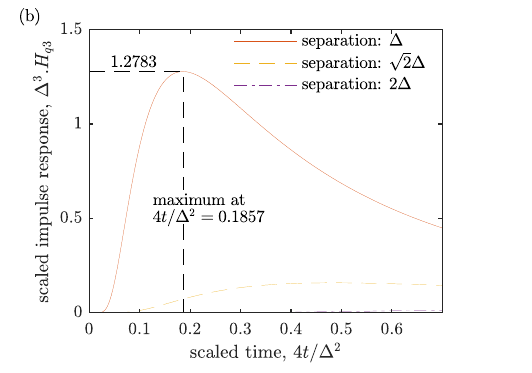}
    \caption{Illustration of pixel-scaled impulse response functions for the three closest elements in the three-dimensional conduction case: (a) flux-to-temperature, $\Delta . H_{g3}$; and (b) temperature-to-flux, $\Delta^3.H_{q3}$.}
    \label{fig:impulse_response_functions}
\end{figure}

\subsection{Transforming from temperature to flux}
\label{sec:temperature-to-flux}

We can also develop analogous expressions for the case of converting from temperature to heat flux. For the case of three-dimensional heat conduction, commencing with \cref{DtoN}, we can write the multi-dimensional component of heat flux is given by 
\begin{align*}
    q_{md}(t,x)  = \int_0^t \int_{-\infty}^{\infty} \int_{-\infty}^{\infty} (g(s,y)-g(s,x))\frac{e^{-((x_1-y_1)^2+(x_2-y_2)^2)/4(t-s)}}{(4\pi)^{3/2}(t-s)^{5/2}} dy_1 dy_2 ds .
\end{align*}
We treat the temperature $g(t,y)$ as uniform over each individual pixel $p$, and each pixel is considered to be square and to represent a distance along the surface of $\Delta$ in each axis direction. Under these conditions, the multi-dimensional component of heat flux for the location of interest is  
\begin{align*}
    q_{md}(t,x)  & = \sum_{p = 1}^{n_p} \int_0^t  (g(s,y_p)-g(s,x)) \int_{y_{2p}-\frac{\Delta}{2} }^{y_{2p}+\frac{\Delta}{2}} \int_{y_{1p}-\frac{\Delta}{2}}^{y_{1p}+\frac{\Delta}{2}} \frac{e^{-((x_1-y_1)^2+(x_2-y_2)^2)/4(t-s)}}{(4\pi)^{3/2}(t-s)^{5/2}} dy_1 dy_2 ds .
\end{align*}
Using the same definition for the integral $I$ as for \cref{eq:Ipixel}, we can write,  
\begin{align}
\label{eq:q_md}
    q_{md}(t,x)& =  \sum_{p = 1}^{n_p} \int_0^t  \underbrace{\left( g(s,y_p)-g(s,x) \right)}_{\textrm{forcing function}} \underbrace{\frac{I(t-s,x_1,y_{1p},\Delta) I(t-s,x_2,y_{2p},\Delta)}{\sqrt{\pi} (4(t-s))^{3/2}}}_{\textrm{pixel-based impulse response function, $H_q$}} ds .
\end{align}
The contribution of any pixel $p$ to the multi-dimensional component of heat flux at location $x$ is obtained through convolution of the temperature difference between these two locations and the impulse response function, which for three-dimensional conduction is
\begin{align}
\label{eq:Hq3}
    H_{q3}(t,x,y_p,\Delta) &= \frac{1}{\sqrt{\pi}} \frac{I(t,x_1,y_{1p},\Delta) I(t,x_2,y_{2p},\Delta)}{(4t)^{3/2}}.
\end{align}
and in the two-dimensional case, the impulse response function is
\begin{align}
\label{eq:Hq2}
    H_{q2}(t,x,y_p,\Delta) &= \frac{2}{\sqrt{\pi}} \frac{I(t,x,y_{p},\Delta)}{(4t)^{3/2}} .
\end{align}

Results are presented in \cref{fig:impulse_response_functions}(b) for transforming from temperature to flux in the case of three-dimensional conduction $H_{q3}$. The necessary impulse response function all have a finite response, and are zero initially. The impulse response from the nearest pixel reaches a maximum at $4t/\Delta^2 = 0.1857$ and has a value of $\Delta^3 H_{g3} = 1.2783$. The corresponding results for the two-dimensional case are $4t/\Delta^2 = 0.2140$ with a maximum value of $\Delta^3 H_{q2} = 1.4404$.

\subsection{Quadratic flux}
\label{sec:discrete_quadratic}

To illustrate the correct recovery of the multi-dimensional heat conduction effects for discrete data, we first consider the case of a quadratic flux distribution in two-dimensions $q = x^2$, since we have complete analytical results for the continuous distribution in this case (\cref{sec:quadraticflux}). For the present discrete analysis, we consider only the portion of the distribution $-6.25 \le x \le 6.25$ discretized into 51 evenly-distributed points. \cref{fig:quadratic_q_and_g}(a) presents the flux over the region of interest with linear interpolation between the 51 points so that $\Delta = 0.25$, making the flux distribution appear continuous in this representation; the inset plot illustrates the discretization. The analytical result for the temperature $g(t,x)$ (\cref{sec:quadraticflux}) is illustrated in \cref{fig:quadratic_q_and_g}(b) over the same region for selected values of time, $t = 0.25$, $0.5$ and $1$. 

\begin{figure}[H]
    \centering
    \includegraphics[width=79mm]{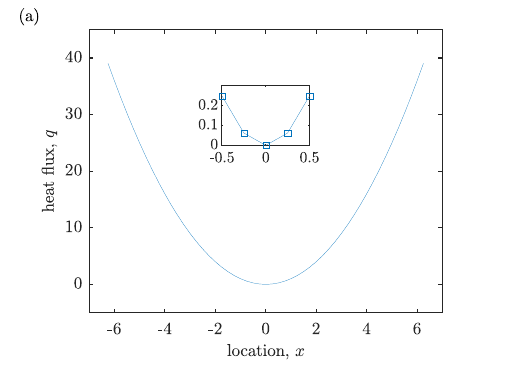}
    \includegraphics[width=79mm]{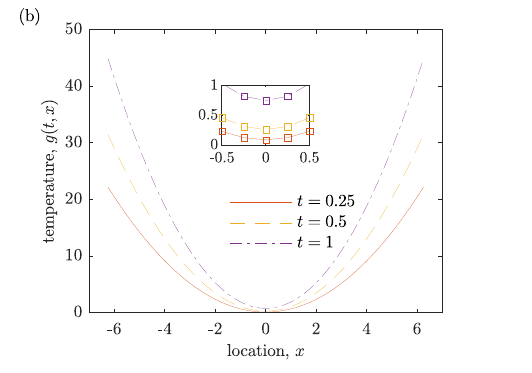} 
    \caption{Time-invariant quadratic heat flux distribution (part a) and illustration of temperature variations produced therefrom (part b).}
    \label{fig:quadratic_q_and_g}
\end{figure}

The maximum time considered herein is $t=1$, and the chosen time step size is $(\Delta^2)(0.2140/8)$. Note that from \cref{sec:temperature-to-flux}, the two-dimensional impulse response function $H_{g2}$ for the nearest element has a maximum at $t = (\Delta^2)(0.2140/4)$. Therefore, the discrete version of the impulse response function used herein will correctly resolve the maximum, and will also include one non-zero point prior to the peak. Correctly resolving the maximum value of the impulse response function for the nearest element has a significant impact on the accuracy of the multi-dimensional heat flux calculation, and including one additional point prior to the peak further improves accuracy. However, further increases in temporal resolution typically provide little, if any further gain in accuracy. 

Evaluation of the one-dimensional component of the heat flux proceeds by treating the local temperature rise $g(t,x)$ using some suitable method, but to evaluate the application of the multi-dimensional result for obtaining heat flux from temperature as presented in \cref{eq:q_md}, temperature differences $g(t,y)-g(t,x)$ are required. For the time being, if we suppose the point of interest is $x = 0$, then the temperature difference results are presented in \cref{fig:quadratic_DT}(a) for several times as a function of location, and in \cref{fig:quadratic_DT}(b) for several locations as a function of time. Results illustrating the one-dimensional component of heat flux obtained from $g(t,x)$ are illustrated in \cref{fig:quadratic_heatflux}(a), and these were obtained using an impulse response filtering method as suggested by Oldfield \cite{Oldfield2008}. Results illustrating the multi-dimensional component of heat flux obtained from the convolution of $g(t,y)-g(t,x)$ and $H_{q2}$ and the summation of results across the 51 elements are presented in \cref{fig:quadratic_heatflux}(b).

\begin{figure}[H]
    \centering
    \includegraphics[width=79mm]{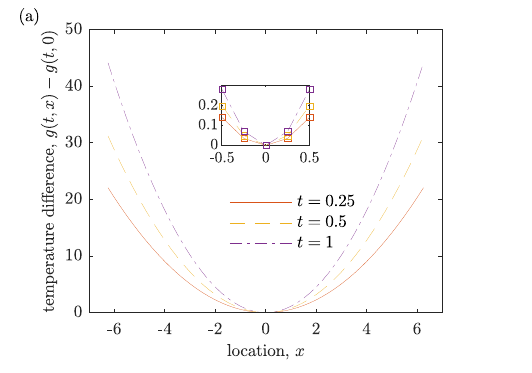}
    \includegraphics[width=79mm]{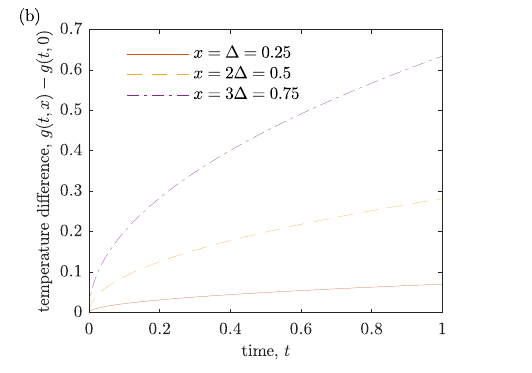} 
    \caption{Temperature difference $g(t,x)-g(t,0)$ for quadratic heat flux case: (a) variation with distance for three times; (b) variation with time for three locations.}
    \label{fig:quadratic_DT}
\end{figure}

\begin{figure}[H]
    \centering
    \includegraphics[width=79mm]{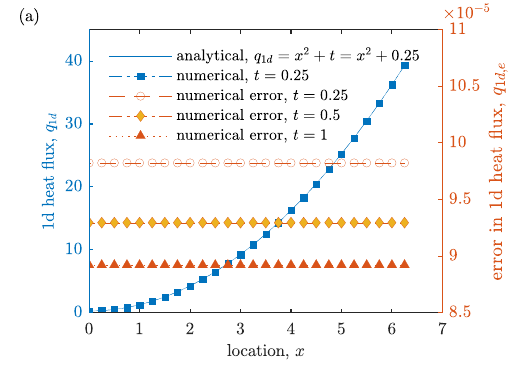} 
    \includegraphics[width=79mm]{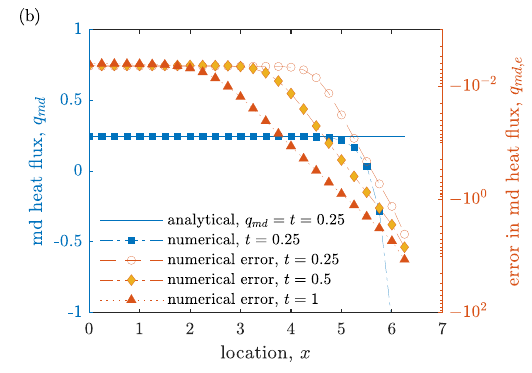} 
    \caption{Components of recovered heat flux for quadratic case: (a)~one-dimensional component; (b)~multi-dimensional component.}
    \label{fig:quadratic_heatflux}
\end{figure}

The errors in the treatment of the one-dimensional component are modest; for the times illustrated in \cref{fig:quadratic_heatflux}(a) the errors in the one-dimensional component are less than $1\times10^{-4}$. Smaller errors in the one-dimensional component can be obtained, depending on the time-resolution of the treatment, but the error values achieved presently are sufficiently small to have a negligible impact on the overall recovery of the heat flux, since the magnitude of the errors in the multi-dimensional component are somewhat larger, as illustrated in \cref{fig:quadratic_heatflux}(b), where the magnitude of the errors close to $x=0$ are around $4\times10^{-3}$. The magnitude of the errors in the multi-dimensional component at any given time are largest at the extremity of the region of interest, ie at $x = \pm 6.25$, and these errors grow with time and become evident closer to the centre of the distribution. The effects of the limited spatial domain being considered are impacting the results in this case.

The overall heat flux result $q = q_{1d} - q_{md}$ is illustrated in \cref{fig:quadratic_qfinal}, where it is demonstrated the time-invariant heat flux that was specified is successfully recovered to a high degree of accuracy, apart from edge effects associated with the finite spatial domain which manifest more prominently at later times and at locations closer to the $x = \pm 6.25$. The error in the overall heat flux arises primarily because of the multi-dimensional treatment, and for the results illustrated in \cref{fig:quadratic_qfinal}, the error in the multi-dimensional component at $(t,x) = (1,0)$ was $4.01\times 10^{-3}$, which is a small quantity relative to the maximum heat flux which was $q_{\text{max}} = x_{\text{max}}^2 \approx 39$. Thus the relative error in the multi-dimensional heat flux amounts to about 0.01\,\%; such a value is insignificant compared to the magnitude of uncertainties that typically arise in transient heat flux experiments. Nevertheless, the error in the multi-dimensional heat flux component can be made smaller if required by using a finer spatial resolution. For example, in the above results, a total of 51 points distributed evenly within $-6.25 \le x \le 6.25$ were used, but if instead we used 11, 21, 41, or 81 points, giving $\Delta = 1.25$, 0.625, 0.3125, or 0.1563, the associated magnitude of error in the multi-dimensional component at $(t,x) = (1,0)$ is calculated to be $109.9\times 10^{-3}$, $28.00\times 10^{-3}$, $6.559\times 10^{-3}$, or $1.333\times 10^{-3}$, respectively. For each different value of $\Delta$, the calculated error results were obtained with the time step defined using $(\Delta^2)(0.2140/8)$.

\begin{figure}[H]
    \centering
    \includegraphics[width=79mm]{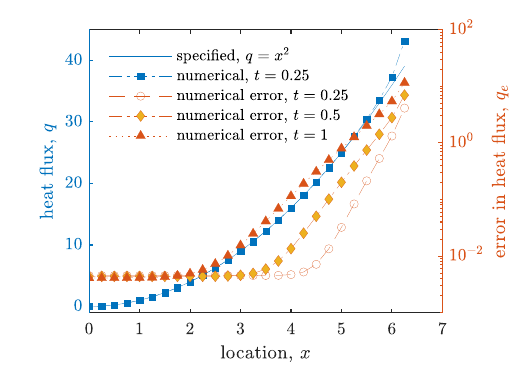} 
    \caption{Heat flux recovered from temperature distribution through numerical treatment for quadratic heat flux case.}
    \label{fig:quadratic_qfinal}
\end{figure}

\subsection{Gaussian flux}
\label{sec:discrete_gaussian}

As a companion to the demonstration of the discretized quadratic flux distribution of \cref{sec:discrete_quadratic}, we perform a similar analysis for the case of a two-dimensional Gaussian distribution, $q(t,x) = e^{-x^2}$. The zone of analysis in a numerical treatment necessarily has a finite size, so having the flux approach zero for sufficiently large $x$ offers an advantage relative to the quadratic distribution. However, unlike the quadratic distribution, we do not have an analytical expression for the variation of surface temperature $g(t,x)$. Nevertheless, we do have an expression that can be numerically integrated to deliver the temperature (see \cref{sec:gaussian_flux_results}) without recourse to the convolution method in \cref{sec:flux-to-temperature}, 
\begin{align*}
    g(t,x) &= \frac{\text{e}^{-x^2}}{\sqrt{\pi}}  \int_0^t \frac{1}{\sqrt{4 s^2 + s}} \text{exp}\left( \frac{4 x^2 s }{4 s +  1} \right)  ds.
\end{align*}

The discretized version of the Gaussian flux distribution that is used in this work is illustrated in \cref{fig:gaussian_q_and_g}(a), although the full region used in the analysis which was actually 
$-6.25 \le x \le 6.25$, as was the case in \cref{sec:discrete_quadratic}. Over the region  $-6.25 \le x \le 6.25$, the distribution was discretized using 51 points, giving $\Delta = 0.25$. Temporal discretization for the Gaussian analysis was the same as that for the quadratic as described in \cref{sec:discrete_quadratic}. Selected results from the numerical integration of the above expression are presented in \cref{fig:gaussian_q_and_g}(b). The temperature distributions $g(t,x)$ are used in the analysis that delivers the one-dimensional conduction effects, but the temperature differences relative to each point of interest $g(t,y)-g(t,x)$ are the quantities needed for the analysis of multi-dimensional conduction effects. A selection of temperature difference results are presented in \cref{fig:gaussian_DT} for the location of interest being $x=0$. 
 
\begin{figure}[H]
    \centering
    \includegraphics[width=79mm]{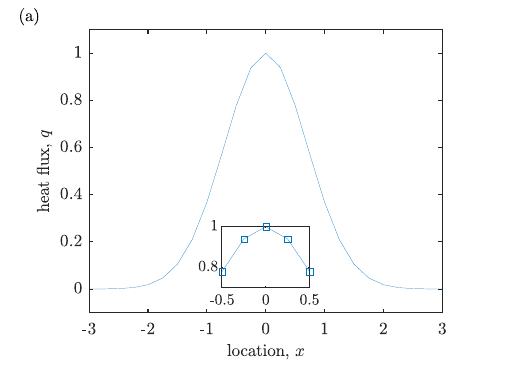}
    \includegraphics[width=79mm]{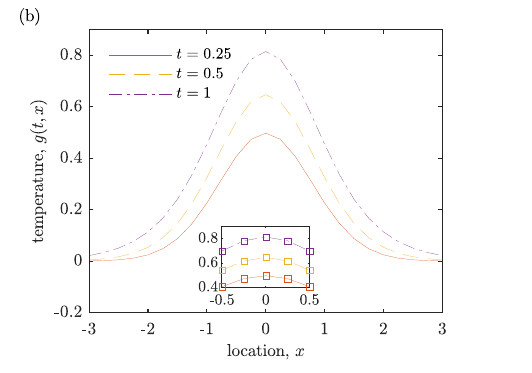} 
    \caption{Time-invariant Gaussian heat flux distribution (part a) and illustration of temperature variations produced therefrom (part b).}
    \label{fig:gaussian_q_and_g}
\end{figure}

\begin{figure}[H]
    \centering
    \includegraphics[width=79mm]{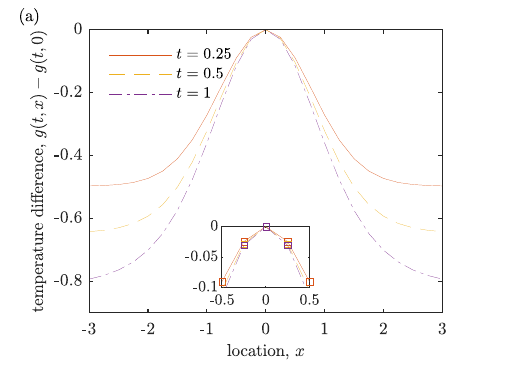}
    \includegraphics[width=79mm]{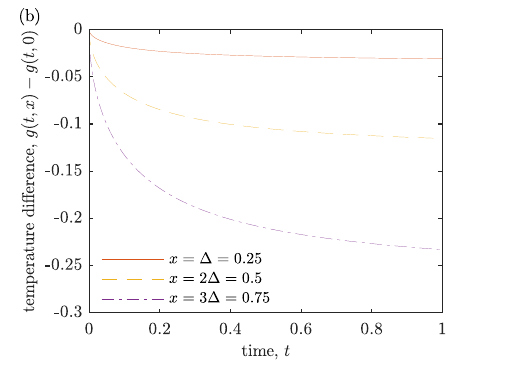} 
    \caption{Temperature difference $g(t,x)-g(t,0)$ for Gaussian heat flux case: (a) variation with distance for 3 times; (b) variation with time for 3 locations.}
    \label{fig:gaussian_DT}
\end{figure}

Results illustrating the one- and multi-dimensional components of the heat flux deduced from treatment of $g(t,x)$ and $g(t,y)-g(t,x)$, respectively, are presented in \cref{fig:gaussian_heatflux}. In the case of the one-dimensional component, \cref{fig:gaussian_heatflux}(a), a version of the impulse response filtering method of Oldfield \cite{Oldfield2008} was again used, and the magnitude of the maximum errors in the numerical results at the first time step (labeled $t=+0$ in the figure) are around $1\times10^{-6}$. The maximum heat flux inferred using the one-dimensional analysis reduces with time, although beyond $x = 1$ the deduced heat flux increases with time, and both of these effects illustrate the limitations of the one-dimensional approach because the actual flux is time-invariant in this case. The multi-dimensional component of the heat flux should precisely compensate for the failings in the one-dimensional component, and turning attention to \cref{fig:gaussian_heatflux}(b) we can see this is likely to be the case. Errors in the numerical treatment of both the one- and multi-dimensional components cannot be assessed separately in the case of the Gaussian, because unlike the quadratic case, we do not have analytical expressions for $q_{1d}$ and $q_{md}$.

\begin{figure}[H]
    \centering
    \includegraphics[width=79mm]{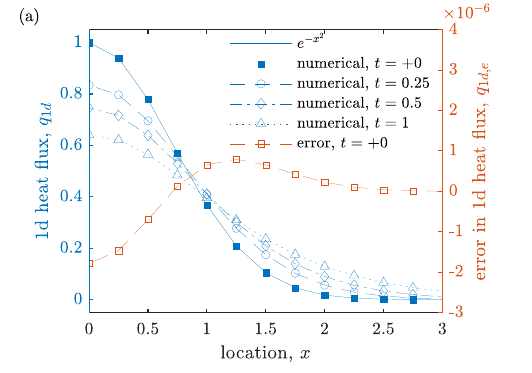} 
    \includegraphics[width=79mm]{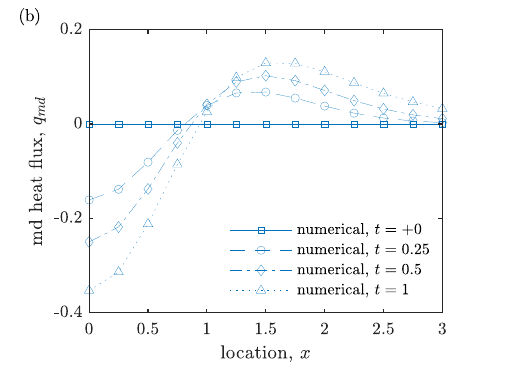} 
    \caption{Components of recovered heat flux for Gaussian case: (a)~one-dimensional component; (b)~multi-dimensional component.}
    \label{fig:gaussian_heatflux}
\end{figure}

The overall results for the heat flux $q = q_{1d} - q_{md}$ are presented in \cref{fig:gaussian_qfinal}, which illustrates that the time-invariant heat flux distribution can be recovered with the maximum magnitude of errors being less than approximately $5\times10^{-3}$. The current discretization was achieved with 51 points distributed across $-6.25 \le x \le 6.26$, giving $\Delta = 0.25$. Such a level of accuracy, which amounts to about 0.5\,\% relative to the maximum heat flux (unity, in this case) would certainly be sufficient for transient heat flux experiments where achieving uncertainties less than 1\,\% is very unlikely. Nevertheless, if higher accuracies are required, then finer levels of discretisation can be used. For example for 41, 81, and 161 points evenly distributed across the interval, the associated magnitude of error in the overall heat flux at $(t,x) = (1,0)$ is calculated to be $5.599\times 10^{-3}$, $1.648\times 10^{-3}$, and $0.3627\times 10^{-3}$, respectively.

\begin{figure}[H]
    \centering
    \includegraphics[width=79mm]{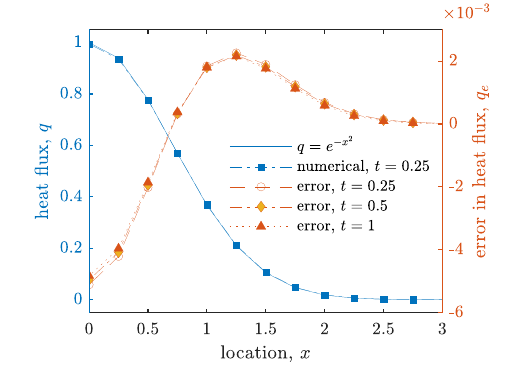} 
    \caption{Heat flux recovered from temperature distribution through numerical treatment for Gaussian heat flux case.}
    \label{fig:gaussian_qfinal}
\end{figure}

\section{Heat flux in a turbulent wedge}

The previous section introduced the approach for treatment of discetized distributions and illustrated application of the method in the case of quadratic and Gaussian distributions, both of which were two-dimensional. In this section we illustrate the application of the treatment to three-dimensional analysis using an empirical model for surface heat flux distribution in a high-speed turbulent flow. We simulate surface temperature data via the convolution method because there are no precise analytical treatments that deliver the surface temperature, unlike the quadratic and Gaussian illustrations in \cref{sec:discretised}. The value of using simulated surface temperature data rather than actual experimental data lies in the fact that there is no experimental uncertainty: we should be able to recover precisely the specified heat flux, apart from numerical errors. 

The chosen heat flux distribution is for a high speed boundary layer on a flat plate with an isolated protuberance that induces transition from laminar to turbulent flow. Examples of turbulent wedges that arise in such configurations are found in the review by Schneider \cite{schneider2008}. The presentation proceeds by: (1) introducing an empirical model for the heat flux distribution; (2) specifying relevant flow properties in a representative hypersonic wind tunnel; (3) specifying the thermal properties of the imagined wind tunnel model; and (4) specifying relevant characteristics of the supposed thermal imaging system. The simulated surface temperature distribution is obtained from the flux-to-temperature results; representative levels of temperature measurement noise are also superimposed on the simulated temperatures. The temperature-to-flux method is then applied to the simulated temperatures to determine the accuracy with which the heat flux distribution can be recovered.   

\subsection{Empirical heat flux model}
\label{sec:emp_wedge_model}

If a protuberance on the surface causes the boundary layer to transition from laminar to turbulent flow, the lateral spreading angle of the wedge $\phi$ (a half-angle, expressed in degrees) within which the flow becomes turbulent is approximated from results in Fischer \cite{fischer1972} using the correlation with the Mach number external to the boundary layer M$_e$ given by
\begin{align*}
    \phi = -0.014\,\textrm{M}_e^3 + 0.30\,\textrm{M}_e^2 -2.5\,\textrm{M}_e + 10 ,
\end{align*}
which appears reasonable for $\textrm{M}_e \le 8$. 

In the laminar portions of the boundary layer, the friction coefficient is taken as
\begin{align*}
    c_{f,L} = \frac{0.664}{\left( \textrm{Re}_x^* \right)^{1/2}}
\end{align*}
and in the turbulent portions, the friction coefficient is taken as (see \cite{White2005})
\begin{align*}
    c_{f,T} = \frac{0.455}{\textrm{log}_e^2\left( 0.06 \textrm{Re}_x^* \right)}.
\end{align*}
In both the laminar and turbulent zones, the Stanton number is evaluated from the friction coefficient via the Reynolds analogy according to
\begin{align*}
    \textrm{St} = \frac{c_f}{2 (\textrm{Pr}^*)^{2/3}}.
\end{align*}
In the above expressions, the Reynolds and Prandtl numbers are given by
\begin{align*}
    \textrm{Re}_x^* = \frac{\rho^* u_e x_s}{\mu^*}
\end{align*}
and
\begin{align*}
    \textrm{Pr}^* = \frac{\mu^* c_p^*}{k^*} ,
\end{align*}
respectively. Asterisk properties in the above expressions are evaluated at the reference temperature $T^*$ given by \cite{Eckert1956}
\begin{align*}
     T^* = T_e + 0.50 (T_w - T_e) + 0.22 (T_{aw} - T_e),
\end{align*}
where $T_e$ is the boundary layer edge temperature, $T_w$ is the wall temperature, and $T_{aw}$ is the adiabatic wall temperature given by
\begin{align*}
    T_{aw} = T_e + \frac{r}{2 c_p} U_e^2 ,
\end{align*}
with $r$ being the recovery factor, which, for laminar flow is $r_L = (\textrm{Pr}^*)^{1/2}$ and for turbulent flow, is taken as $r_T = (\textrm{Pr}^*)^{1/3}$. The heat transfer at the surface $q_s$ in either the laminar or turbulent portions of the boundary layer is calculated from the Stanton number according to  \begin{align*}
  q_s = \textrm{St} \rho_e u_e c_p \left( T_{aw} - T_w \right).
\end{align*}

To define the heat transfer in the transition from laminar to turbulent flow, the lateral spreading angle $\phi$ is taken as defining the streamwise location of the nominal start of the transition process $x_t$, with the center of the transition being defined by
\begin{align*}
    x_c = x_t + \frac{x_{TL}}{2} ,
\end{align*}
where $x_{TL}$ is the nominal transition length that is taken to be 25\,\% of the streamwise distance to the transition onset $x_{t,0}$ that would occur in the absence of the protuberance,
\begin{align*}
    x_{TL} = 0.25 x_{t,0} .
\end{align*}
For smooth surfaces, a correlation for the variation of transition onset with Mach number external to the boundary layer on smooth sharp cones in supersonic and hypersonic wind tunnels is given by  
\begin{align*}
    \textrm{log}_{10} \left( \textrm{Re}_{x_{t,0}} \right) = 0.0116\,\textrm{M}_e^2 - 0.100\,\textrm{M}_e + 6.63 ,
\end{align*}
based on the data collated in \cite{Beckwith1975, BeckwithBertram1972}.  In the above expression, the Reynolds number at transition onset is defined in terms of the flow properties external to the boundary layer,
\begin{align*}
    \textrm{Re}_{x_{t,0}} = \frac{\rho_e u_e x_{t,0}}{\mu_e} .
\end{align*}
The hyperbolic tangent function is used to generate a smooth variation in the time-averaged heat flux throughout the transition from laminar to turbulent flow. The function is scaled so that the 1\,\% to 99\,\% rise from laminar to turbulent flux values occurs over the nominal transition length, $x_{TL}$. That is,
\begin{align*}
     q_{s} = q_{s,L} +  \left( 1+\tanh \frac{x_s-x_{c}}{x_\textrm{ref}} \right) \frac{q_{s,T} - q_{s,L}}{2},
\end{align*}
where $x_\textrm{ref}$ is a reference length given by,
\begin{align*}
     x_\textrm{ref} = \frac{x_{TL}}{2 \tanh^{-1}(0.98)} .
\end{align*}

\subsection{Specified heat flux}
\label{sec:windtunnel_heatflux}

The flow conditions selected for this work correspond to a Mach 6 flow of dry air generated isentropically from stagnation pressure and temperature of 900\,kPa and 550\,K, respectively. The ratio of specific heats for the air is assumed constant, $\gamma = 1.4$. Sutherland's law is used for the viscosity and conductivity of the air. We assume a flat plate is aligned with the Mach~6 flow direction with a protuberance located 140\,mm downstream from its leading edge. The temperature of the flat plat is treated as constant, $T_w = 293$\,K. Although the plate temperature will actually increase when exposed to the heat load, it is assumed that the flow duration is sufficiently short for the change in $T_w$ to be a small fraction the temperature difference $T_{aw}-T_w$ that drives the heat flux.   According to the model described in \cref{sec:emp_wedge_model}, the distribution of heat flux for these conditions is illustrated in \cref{fig:wedge_heatflux_image}(a). 

\begin{figure}[hbt!]
    \centering
    \includegraphics[trim={0 1.0cm 0 1.5cm},clip,width=79mm]{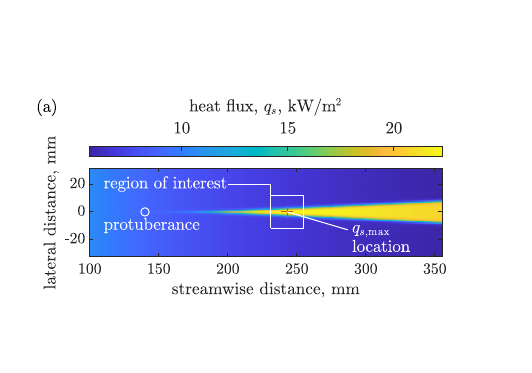} \\
    \includegraphics[trim={0 0.0cm 0 0cm},clip,width=79mm]{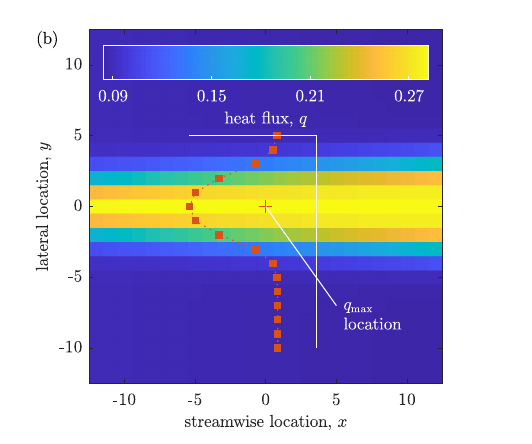} 
    \caption{Simulated heat flux distribution: (a) heat flux $q_s$ (engineering units, kW/m$^2$) within turbulent wedge; (b) scale-invariant heat flux $q$ within the region of interest.}
    \label{fig:wedge_heatflux_image}
\end{figure}

\subsection{Temperature rise from heat flux}
\label{sec:plate_temperature_rise}

The plate is assumed to be made of PEEK (polyetheretherketone) which is frequently used in high speed flow ground-testing applications (e.g., see Zhang et al. \cite{zhang2021heat}). PEEK has a relatively high emissivity, enabling surface temperature measurement via imaging with a thermal IR camera.  The following thermal properties are assumed for the PEEK: $\rho = 1300$\,kg/m$^3$; $ c = 1100$\,J/kgK; and $k = 0.27$\,W/Km, giving a value for thermal diffusivity of $\alpha = 1.89\times10^{-7}$\,m$^2$/s.

To acquire the plate surface temperature, it is assumed that an IR camera with a suitable lens system produces the view of the zone illustrated in \cref{fig:wedge_heatflux_image}(a) with 256 $\times$ 64 pixels. The lens arrangement is such that the surface viewing area of each pixel is 1\,mm square. The region of primary interest is the $25\times25$-pixel zone centered on the location of the maximum turbulent heat flux $q_{s,max}$ as illustrated in \cref{fig:wedge_heatflux_image}(a). The projected pixel size of 1\,mm provides the length scale $L = 1$\,mm, and the scale-invariant heat flux is  
\begin{equation*}
    q =  \frac{q_s L}{k \, T_w}.
\end{equation*}
Scale-invariant values of heat flux $q$ produced in this manner are illustrated in \cref{fig:wedge_heatflux_image}(b) for the $25 \times 25$ pixel region of interest centered on the maximum heat flux value. The maximum value of heat flux in this case is $q_{\text{max}} = 0.2818$. For analysis of actual experimental data, conversion to scale-invariant quantities is not an imperative, but given the convenience afforded and the extensive use of scale-invariant quantities in the earlier sections of this work, we continue with scale-invariant properties, but refer back to scale-specific properties when required.

Heat flux results in \cref{fig:wedge_heatflux_image}(b) were converted to temperature distributions using the convolution method described in \cref{sec:flux-to-temperature}. In the three-dimensional case, the impulse response function for converting from heat flux to temperature $H_{g3}$ has a maximum at $4t/\Delta^2 = 0.4304$ for the shortest separation of elements (\cref{fig:impulse_response_functions}), and for accurate analysis, a time resolution that at least matches this value is needed. However, because we subsequently want to convert from the temperature back into heat flux and the impulse response function for doing so ($H_{q3}$) reaches a maximum at an even shorter time, $4t/\Delta^2 = 0.1857$ (\cref{fig:impulse_response_functions}), we use a time step that allows three non-zero values in the temperature-to-flux impulse response prior to the peak in $H_{q3}$. (The choice of three non-zero points is somewhat arbitrary.) The time step size in this case is therefore $t = 0.1857/16 = 0.011606$. The corresponding time step size expressed in seconds is $t_s = t L^2/\alpha = 0.0614$\,s, requiring the IR camera in this case to have a frame rate of about 16\,Hz, which, for the specified number of pixels, is not demanding for current IR camera technology. 

Results from the conversion from flux to temperature are illustrated in \cref{fig:wedge_g_montage} for selected times up to $t=4$. In terms of actual experiment duration, the value $t=4$ corresponds to $t_s = t L^2/\alpha = 21.2$\,s, which is longer than the experiment reported by Zhang et al. \cite{zhang2021heat}, but is not an excessive time period and could be achieved in other blow-down facilities. Results in \cref{fig:wedge_g_montage} were obtained by first calculating the one-dimensional component of the temperature rise $g_{1d}$ as specified in \cref{NtoD}, and this was achieved at each pixel using the impulse response filtering method of Oldfield \cite{Oldfield2008}. The multi-dimensional component $g_{md}$ at each pixel was calculated by the convolution-then-summation method as described in \cref{sec:flux-to-temperature}. The net result was then produced, $g = g_{1d} + g_{md}$.  It is noted that the simulated temperatures in regions near the left and right vertical edges of the images are low relative to the values along the centre line of each image. This effect arises because the flux-to-temperature analysis considers only the $25 \times 25$ region of pixels illustrated; the analysis effectively treats surfaces outside of this region of interest as receiving no heat flux, causing the simulated temperatures to fall-off unrealistically near the edges. 

\begin{figure}[hbt!]
    \centering
    \includegraphics[trim={0cm 3.3cm 0 0.5cm},clip,width=79mm]{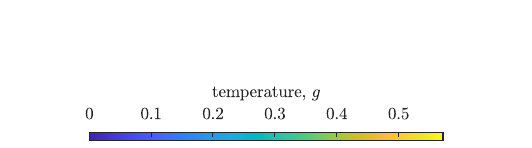} \\
    \includegraphics[width=40mm]{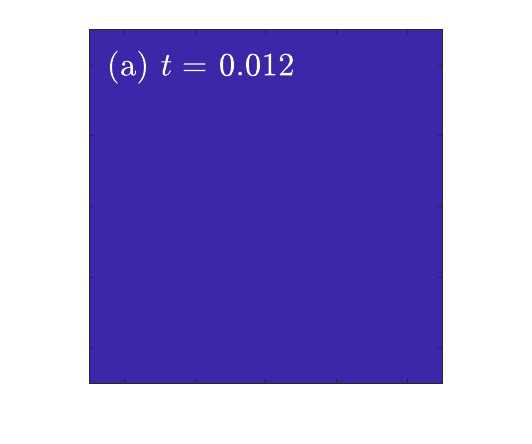}
    \includegraphics[width=40mm]{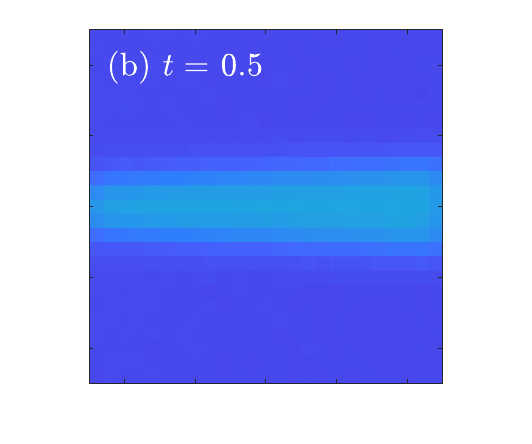} 
    \includegraphics[width=40mm]{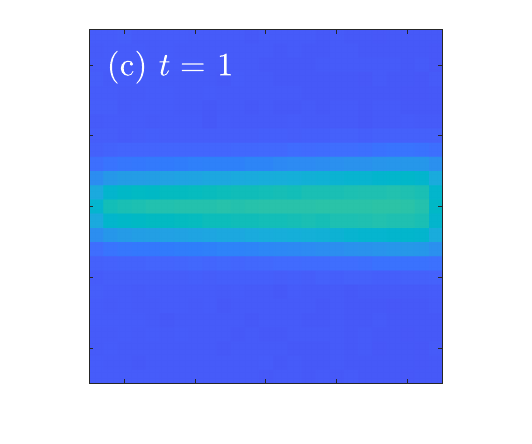} 
    \includegraphics[width=40mm]{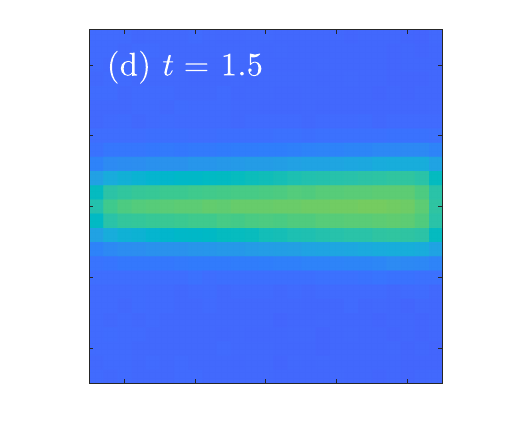} 
    \includegraphics[width=40mm]{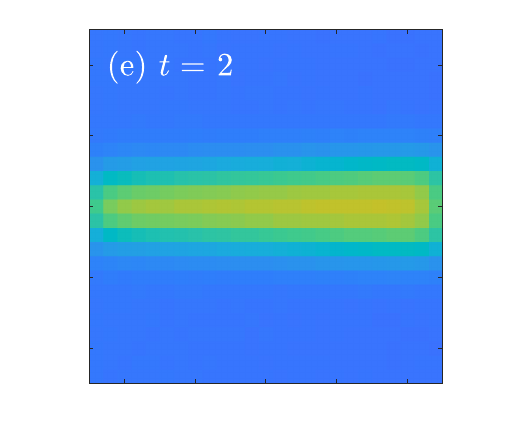} 
    \includegraphics[width=40mm]{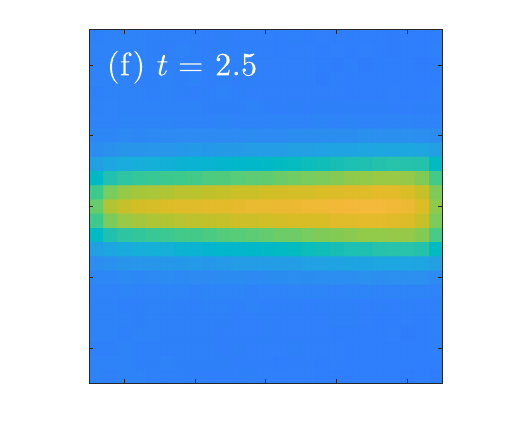} 
    \includegraphics[width=40mm]{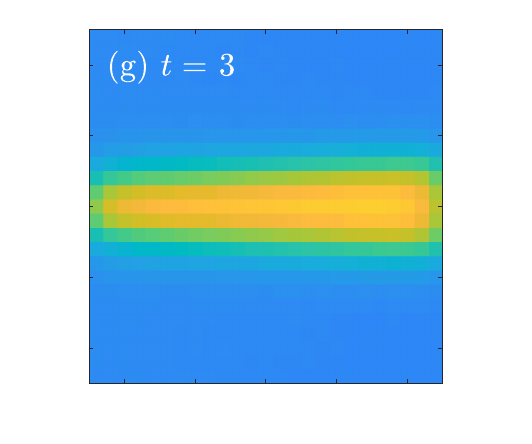} 
    \includegraphics[width=40mm]{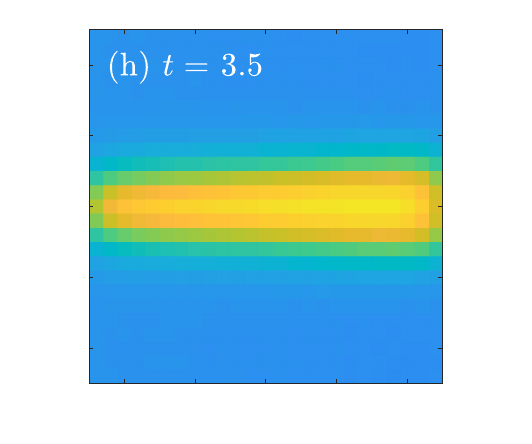} 
    \includegraphics[width=40mm]{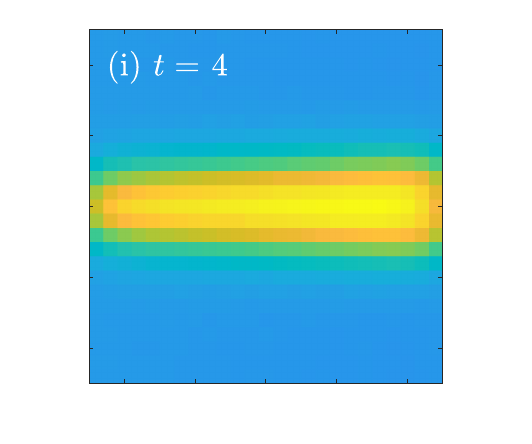} 
    \caption{Simulated temperature development within the region of interest.}
    \label{fig:wedge_g_montage}
\end{figure}

\cref{fig:wedge_g_history} illustrates the contributions from the one-dimensional and the multi-dimensional components of the analysis for the treatment of the pixel at the centre of the region of interest, coinciding with the location of the maximum heat flux. The multi-dimensional component is zero initially, so the combined result initially follows the one-dimensional component. At $t = 4$, the multi-dimensional component is $g_{md} = -0.078$ and the one-dimensional component is $g_{1d} = 0.636$, giving a combined result of $g = 0.558$ at this time. The combined temperature result is lower than the one-dimensional result because temperatures at all other locations are necessarily lower than at this point of maximum heat flux, so there will be effects that conduct heat away from this point, lowering its temperature relative to the one-dimensional result. 

Prior to conversion from temperature back into heat flux, a noise component $g'$ is added to the simulated temperature results, $g = g_{1d} + g_{md} + g'$. All actual temperature measurements will have some noise, and in this case, the simulated noise $g'$ is Gaussian with a standard deviation of $1\times10^{-3}$, a sample of which is included as the inset in \cref{fig:wedge_g_history}. The temperature simulations  illustrated in \cref{fig:wedge_g_montage} are free of such noise, but prior to the deduction of heat flux, noise of the aforementioned magnitude was added to the simulated temperature history for each pixel within the region of interest. For a surface temperature of 293\,K in this case, the standard deviation of the Gaussian noise represents a value of 0.293\,K. The temporal noise of IR camera systems is frequently described using a Noise Equivalent Temperature Difference (NETD) and values lower than 100\,mK are routinely achieved. Therefore, the noise level applied to the present simulated temperatures (standard deviation of 293\,mK) exceeds the noise values that should be achievable in practice by a significant margin.  

\begin{figure}[hbt!]
    \centering
    \includegraphics[width=79mm]{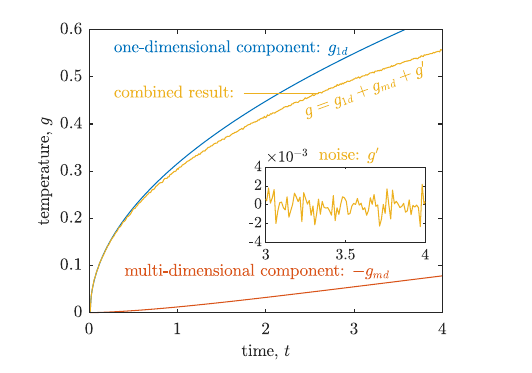} 
    \caption{Temperature history at the point of maximum heat flux, in the centre of the region of interest.}
    \label{fig:wedge_g_history}
\end{figure}

\subsection{Heat flux from temperature rise}

Heat flux within the region of interest was obtained from the temperature simulations with noise included, $g = g_{1d}+g_{md}+g'$, and results at selected times are presented in \cref{fig:wedge_q_montage}. The added noise on the temperature history for each pixel introduces significant fluctuations into the deduced heat flux. However, the overall effect is that the required heat flux, which is time-invariant according to simulated experimental conditions, can be recovered with reasonable accuracy if sufficient low-pass filtering of the results is applied. The mean value of the heat flux at the center pixel in \cref{fig:wedge_q_montage} over the period up to $t=4$ is 0.280, whereas the specified value at the center of this pixel is $q_\text{max} = 0.282$, an error of less than 1\,\%.

\cref{fig:wedge_q_history} presents the detailed history at the center pixel. The value of heat flux from the one-dimensional analysis $q_{1d}$ decreases with time, but the magnitude of the multi-dimensional component $q_{md}$ increases such that the combined result $q = q_{1d}-q_{md}$ is essentially constant, apart from the noise. 
The broken lines in \cref{fig:wedge_q_history} provide the heat flux results in the absence of noise in the temperature histories. 
The noise introduced to the temperature history primarily manifests in the heat flux through the one-dimensional component of the analysis $q_{1d}$; the multi-dimensional component $q_{md}$ delivers only a minor contribution to the overall noise. 
In the present analysis, the standard deviation of the noise that is manifested in the multi-dimensional heat flux component relative to the surface temperature noise is $q'_{md}/g' = 0.177$, whereas for the one-dimensional component it is $q'_{1d}/g' = 8.82$. Clearly, the multi-dimensional component of the heat flux contributes very little to the noise in the overall result.

\begin{figure}[hbt!]
    \centering
    \includegraphics[trim={0cm 3.3cm 0 0.5cm},clip,width=79mm]{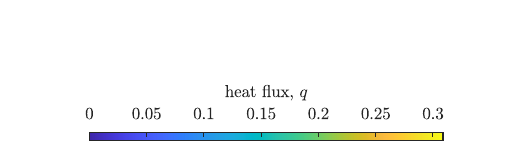} \\
    \includegraphics[width=40mm]{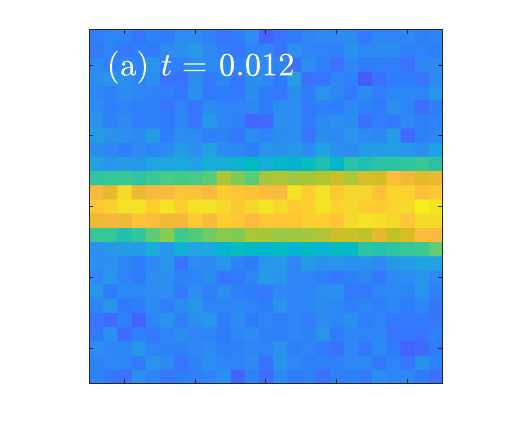}
    \includegraphics[width=40mm]{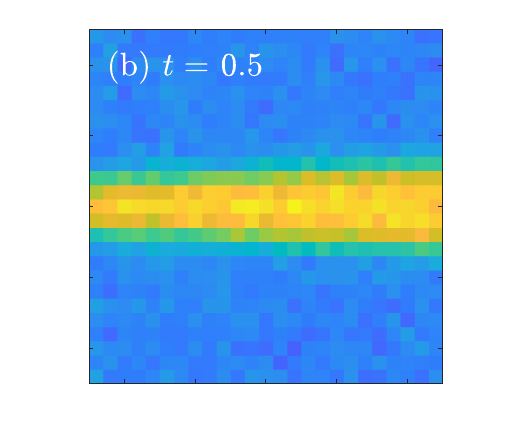} 
    \includegraphics[width=40mm]{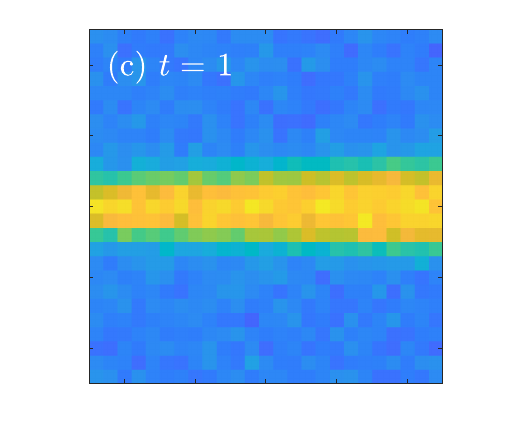} 
    \includegraphics[width=40mm]{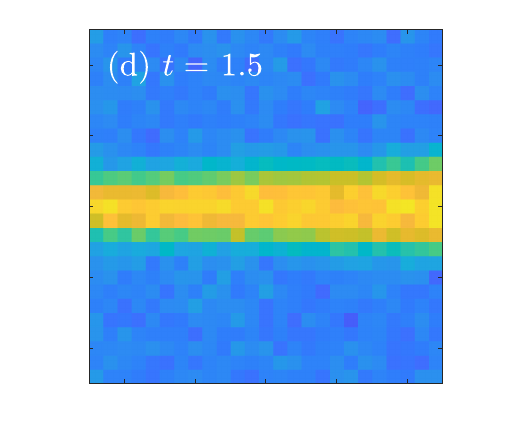} 
    \includegraphics[width=40mm]{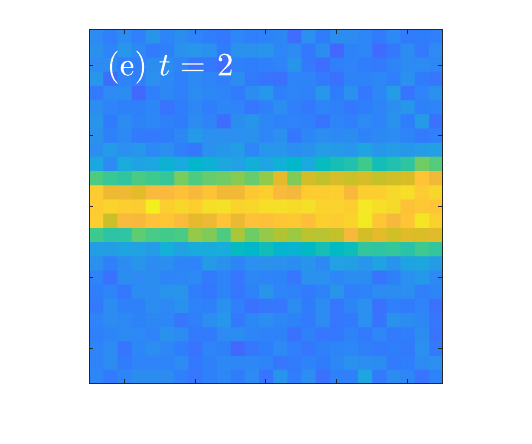} 
    \includegraphics[width=40mm]{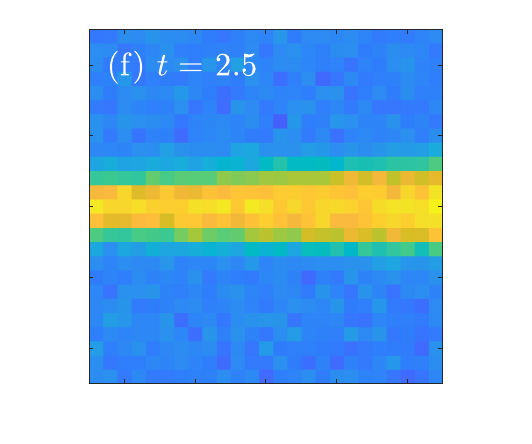} 
    \includegraphics[width=40mm]{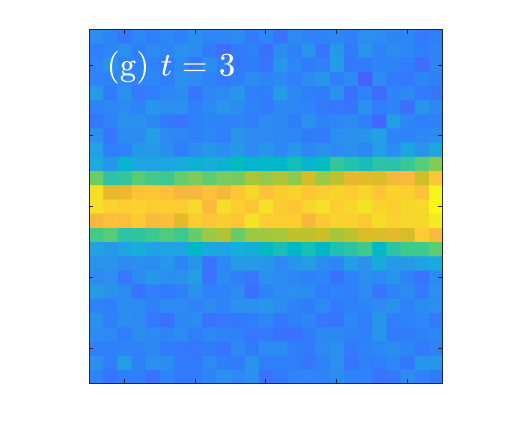} 
    \includegraphics[width=40mm]{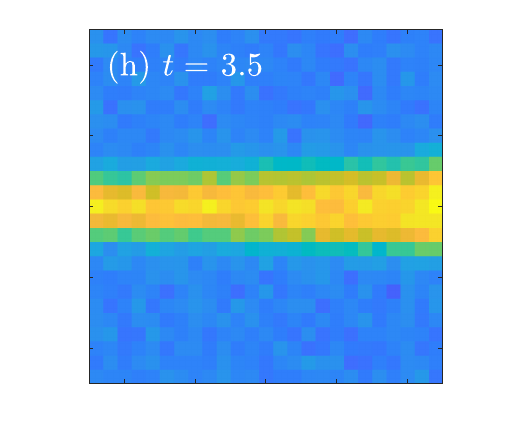} 
    \includegraphics[width=40mm]{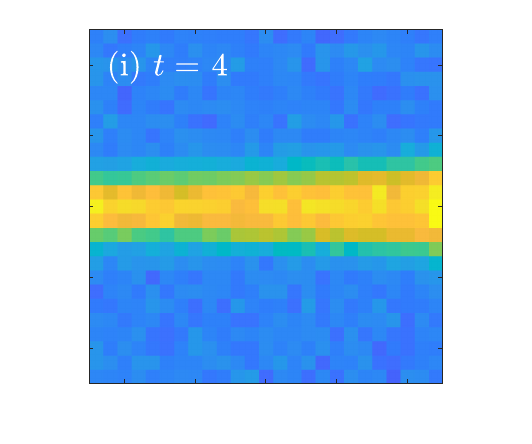} 
    \caption{Deduced heat flux development within the region of interest.}
    \label{fig:wedge_q_montage}
\end{figure}

\begin{figure}[hbt!]
    \centering
    \includegraphics[width=90mm]{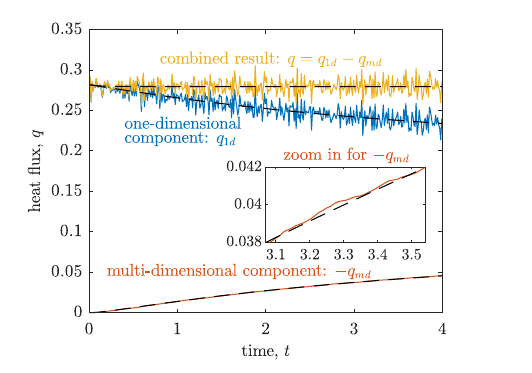} 
    \caption{Heat flux history at the point of maximum heat flux. Broken lines indicate results in the absence of temperature-history noise.}
    \label{fig:wedge_q_history}
\end{figure}

\section{Conclusion}

Multi-dimensional conduction effects will be present in transient heat transfer experiments that include regions with significant variations in surface heat flux, and such effects will generally become larger as the experiment duration increases.  Accurate treatments for multi-dimensional conduction effects have been reported in the literature, but such methods have not been universally adopted by experimentalists. The lack of up-take may be due to perceived complexity of the accurate multi-dimensional treatment, or perhaps the insidious nature of the multi-dimensional errors that typically grow monotonically with time. In any case, experimentalists are familiar with one-dimensional transient heat flux treatments, so we have presented expressions that distinguish the one-dimensional and multi-dimensional components. The approach we advocate enables application of any preferred one-dimensional analysis with the multi-dimensional component treated as a correction term. Therefore, the only challenge for experimentalists is implementing a treatment for the multi-dimensional component.

We suggest a treatment for the multi-dimensional component that is both simple and accurate for modest levels of discretization. Our primary simplification is that the measured surface temperature is considered constant within the spatial element over which the temperature measurement is taken. This allows the multi-dimensional conduction component to be determined by first convolving temperature differences with respective impulse response functions, and then summing the results across all pixels within the region of interest. The necessary impulse response functions are always finite; there are no singularities encountered in the treatment of the multi-dimensional component. The method enables recovery of the heat flux values with an accuracy of better than 1\,\% using around 50 surface elements in the case of representative two-dimensional heat flux variations, and around 600 surface elements in the case of an illustrative three-dimensional problem. Higher levels of accuracy can be achieved if required, using finer spatial discretization.

The present work has treated the semi-infinite case but the method can, in principle, be applied to any geometry. Impulse response functions for the surface elements (pixels) at which temperature measurements are obtained will need to be defined. In general, impulse response functions are obtained by integrating heat kernel results across the surface element on which temperature measurements are obtained. Approximations for the required heat kernels can be obtained either from truncated infinite sums in the case of geometries with symmetry, or from finite element numerical methods in more complicated cases. 

\bibliography{q_from_T}

\appendix

\section{Some Justifying Computations}\label{justify}
The formula \eqref{DtoN} expresses the flux ($q$) at the boundary $\partial \Omega$ in terms of the surface temperature ($g$). In this Appendix, we provide some auxillary computations to justify a number of the steps that were taken to find this formula. We will make the following assumption on the temperature profile $u$ on $\Omega$: for each multi-index $\gamma=(\gamma_1,\cdots,\gamma_n)\in \mathbb{N}^n$, there exists a $C>0$ so that 
\begin{align}\label{ugrowth}
    \left|D^{\gamma}u(t,x)\right|\le C e^{C\left|x\right|} \ \text{for each} \ t\in [0,T_{\text{max}}]. 
\end{align}
This multi-index notation $D^{\gamma}$ tells us how many times to differentiate $u$, and in which variables. For example $D^{(1,0,0,\cdots,0)}u$ means the first derivative of $u$ with respect to the first variable $x_1$. Thus, condition \eqref{ugrowth} simply means that any partial derivatives of $u$ do not grow faster than exponentially at spatial infinity. 
Clearly the temperature in most real-world applications will satisfy this sub-exponential growth assumption.

\begin{lemma}\label{lem:vdformula}
Fix any $t>0$. 
    For $v$ given as in \eqref{vsoln}, formula \eqref{vdinside} holds for all $x\in \Omega$ including for $x\in \partial \Omega$. Furthermore, $\frac{\partial v}{\partial x_n}(t,x)$ is continuous, all the way up to the boundary. 
\end{lemma}
\begin{proof}
Let $w(t,x)=\left(\frac{\partial }{\partial s}-\Delta\right)\tilde{u}(t,x)$. 
    First, observe that the expression 
    \begin{align*}
       V(t,x)=\int_0^t \int_{\Omega}\frac{\partial k_D(t-s,x,y)}{\partial x_n}w(s,y)dyds 
    \end{align*}
    is defined for each $t>0$ and $x\in \Omega$, including the boundary $\partial \Omega$, i.e., these integrals converge. To see this, we can use the change of variables $\frac{x-y}{\sqrt{t-s}}=z$, giving 
    \begin{align*}
        \int_0^t \int_{\mathbb{R}^n} \left|\frac{\partial k_D(t-s,x,y)}{\partial x_n}w(s,y)\right|dyds\le \int_0^t\frac{C}{(4\pi)^{n/2}\sqrt{t-s}}\int_{\mathbb{R}^n}\left|z_n\right|\text{exp}\left(C\left|x-\sqrt{t-s}z\right|-\frac{\left|z\right|^2}{4}\right) dzds;
    \end{align*}
    this integral clearly converges for each fixed $t$ and $x$. 

    Now we want to show that $V=\frac{\partial v}{\partial x_n}$. 
    Define a cut-off function $\eta:[0,\infty)\to [0,1]$ so that $\eta(r)=0$ for $r\le 1$, and $\eta(r)=1$ for $r\ge 2$. Now consider the function
    \begin{align*}
        v_{\epsilon}(t,x)=\int_0^t \eta \left(\frac{t-s}{\epsilon}\right)\int_{\Omega}k_D(t-s,x,y)w(s,y)dyds;
    \end{align*}
    we are deliberately removing the $x=y$, $t=s$ singularity. Lebesgue's dominated convergence theorem implies that for each fixed $x$ and $t$, $\lim_{\epsilon\to 0}v_{\epsilon}(t,x)=v(t,x)$. Another application of the dominated convergence theorem gives that for each $\epsilon>0$, we have  
    \begin{align*}
         \frac{\partial v_{\epsilon}}{\partial x_n}=\int_0^t \eta\left(\frac{t-s}{\epsilon}\right)\int_{\Omega}\frac{\partial k_D(t-s,x,y)w(s,y)}{\partial x_n}dyds
    \end{align*}
    so that 
    \begin{align*}
    \frac{\partial v_{\epsilon}}{\partial x_n}-V(t,x)=\int_{t-2\epsilon}^{t}\left(\eta\left(\frac{t-s}{\epsilon}\right)-1\right)\int_{\Omega}\frac{\partial k_D(t-s,x,y)w(s,y)}{\partial x_n}dyds. 
    \end{align*}
    Using the same change of variables $z=\frac{x-y}{\sqrt{t-s}}$ as before, we conclude that $\frac{\partial v_{\epsilon}}{\partial x_n}$ converges to $V(t,x)$ as $\epsilon\to 0$, uniformly on compact subsets for $x$ and $t$. The uniform limit theorem implies that $V(t,x)$ is continuous (since $\frac{\partial v_{\epsilon}}{\partial x_n}$ is obviously continuous for each $\epsilon>0$). Finally, we show that $\frac{\partial v(t,x)}{\partial x_n}=V(t,x)$. For this computation, let $x=x_b-\nu x_n$, where $x_b=(x_1,x_2,\cdots,x_{n-1},0)$ denotes the projection onto the boundary. Using the fundamental theorem of calculus, we have 
    \begin{align*}
        v(t,x)&=\lim_{\epsilon\to 0}v_{\epsilon}(t,x)\\
        &=\lim_{\epsilon\to 0}\left(v_{\epsilon}(t,x_b)+\int_0^{x_n}\frac{\partial v_{\epsilon}}{\partial x_n}(t,x_b-\nu s)ds\right)\\
        &=v(t,x_b)+\lim_{\epsilon\to 0}\int_0^{x_n}\frac{\partial v_{\epsilon}}{\partial x_n}(t,x_b-s\nu)ds.
    \end{align*}
    Uniform convergence implies 
    \begin{align*}
        v(t,x)&=v(t,x_b)+\int_0^{x_n}V(t,x_b-s\nu)ds.
    \end{align*}
    The fundamental theorem of calculus then implies that $\frac{\partial v}{\partial x_n}$ exists and coincides with $V$, as required. 
\end{proof}
\begin{lemma}
   If $g$ is smooth, then for each fixed $t>0$ and $x\in \Omega$, we have $\lim_{\tau\to t}(g(\tau,x)-g(t,x))\int_{\Omega}p(t-\tau,x,y)dy=0$.
\end{lemma}
\begin{proof}
   A direction computation gives $\int_{\Omega}p(t-\tau,x,y)dy=-\frac{1}{\sqrt{\pi(t-\tau)}}$. The result then follows from the observation that for each $\tau,t,x$, we have $g(\tau,x)-g(t,x)=\frac{\partial g}{\partial t}(\tau^*,x)(\tau-t)$ for some $\tau^*\in [\tau,t]$; this is a consequence of the mean value theorem.
\end{proof}

\begin{lemma}
   If $x\in \partial \Omega$, then  $\lim_{\tau\to t}\left(\int_{\Omega}p(t-\tau,x,y)(\tilde{u}(\tau,y)-g(\tau,x))dy-\frac{\partial \tilde{u}}{\partial \nu_x}(t,x)\right)=0$. 
\end{lemma}
\begin{proof}
   Since $\tilde{u}(\tau,x)=g(\tau,x)$, Taylor's theorem and the growth condition \eqref{ugrowth} tell us that $\tilde{u}(\tau,y)-g(\tau,x)=D\tilde{u}(\tau,x)+O(\left|y-x\right|^2)$, where the second term satisfies $\left|O(\left|y-x\right|^2)\right|\le C_1\left|y-x\right|^2e^{C_1\left|y-x\right|}$ for some constant $C_1$.    
   In the semi-infinite case, we therefore have 
\begin{align*}
    \int_{\Omega}p(t-\tau,x,y)(\tilde{u}(\tau,y)-\tilde{g}(\tau,x))dy&=\int_{\Omega} p(t-\tau,x,y)(D\tilde{u}(\tau,x)\cdot (y-x)+O(\left|y-x\right|^2))dy\\
    &=\int_{\Omega}-\frac{y_n}{(t-\tau)(4\pi (t-\tau))^{\frac{n}{2}}}\text{exp}\left(\frac{-\left|x-y\right|^2}{4(t-\tau)}\right)(D\tilde{u}(\tau,x)\cdot (y-x)+O(\left|y-x\right|^2))dy\\
    &=\frac{\partial \tilde{u}}{\partial \nu_x}(\tau,x)\int_{\Omega}\frac{y_n^2}{(t-\tau)(4\pi (t-\tau))^{\frac{n}{2}}}\text{exp}\left(\frac{-\left|x-y\right|^2}{4(t-\tau)}\right)\\
    &+\sum_{j=1}^{n-1}\frac{\partial u}{\partial x_j}(\tau,x)\int_{\Omega}-\frac{y_n(x_j-y_j)}{(t-\tau)(4\pi (t-\tau))^{\frac{n}{2}}}\text{exp}\left(\frac{-\left|x-y\right|^2}{4(t-\tau)}\right)dy\\
    &+\int_{\Omega}\frac{y_n}{(t-\tau)(4\pi (t-\tau))^{\frac{n}{2}}}\text{exp}\left(\frac{-\left|x-y\right|^2}{4(t-\tau)}\right)O(\left|y-x\right|^2)dy.
\end{align*}
A direct computation demonstrates that the second term vanishes. The third term does not necessarily vanish, but by performing the usual $z=\frac{x-y}{\sqrt{t-\tau}}$ change of variables, we see that the third term vanishes in the $\tau\to t$ limit. 
The first term is 
\begin{align*}
    \frac{\partial \tilde{u}}{\partial \nu_x}(\tau,x)\int_{\Omega}\frac{y_n^2}{(t-\tau)(4\pi (t-\tau))^{\frac{n}{2}}}\text{exp}\left(\frac{-\left|x-y\right|^2}{4(t-\tau)}\right)&= \frac{\partial \tilde{u}}{\partial \nu_x}(\tau,x)\int_{\Omega}\frac{z_n^2}{(4\pi )^{\frac{n}{2}}}\text{exp}\left(\frac{-\left|z\right|^2}{4}\right)dz\\
    &=\frac{\partial \tilde{u}}{\partial \nu_x}(\tau,x)\int_{0}^{\infty}\frac{z_n^2}{\sqrt{4\pi}}\text{exp}\left(\frac{-z_n^2}{4}\right)dz_n\\
    &=\frac{\partial \tilde{u}}{\partial \nu_x}(\tau,x).
\end{align*} 
\end{proof}
\begin{lemma}
    $\lim_{\tau\to t}\int_0^t (g(s,x)-g(t,x))\int_{\partial \Omega} \nabla_{\nu_y}p(t-s,x,y)dyds$ exists, even though we are integrating over a singularity. 
\end{lemma}
\begin{proof}
    We have already seen that
    $ \int_{\partial \Omega}\nabla_{\nu_y}p(t-s,x,y)dy=\frac{1}{(t-s)\sqrt{4\pi(t-s)}}$. 
    Therefore, assuming $g$ is smooth and using the mean value theorem, 
    \begin{align*}
        \left|g(s,x)-g(t,x)\right|\int_{\partial \Omega} \nabla_{\nu_y}p(t-s,x,y)dy\le \frac{\sup_{v\in [0,t]}\left|\frac{\partial g(v,x)}{\partial t}\right|}{\sqrt{4\pi (t-s)}}.
    \end{align*}
    This is clearly integrable for $s\in [0,t]$. 
\end{proof}
\begin{lemma}
    $\int_{0}^{t}\int_{\partial \Omega}(g(s,y)-g(s,x))\nabla_{\nu_y}p(t-s,x,y)dyds$ exists. 
\end{lemma}
\begin{proof}
Using the change of variables $z=\frac{y-x}{\sqrt{t-s}}$, so $dz=\frac{dy}{(t-s)^{(n-1)/2}}$, we find 
    \begin{align*}
     &\int_{0}^{t}\int_{\partial \Omega}(g(s,y)-g(s,x))\nabla_{\nu_y}p(t-s,x,y)dyds\\
     &=\int_0^t \int_{\mathbb{R}^{n-1}}\frac{e^{-\left|z\right|^2/4}}{(t-s)^{\frac{3}{2}}(4\pi)^{\frac{n}{2}}}(g(s,x+\sqrt{t-s}z)-g(s,x))dzds\\
     &=\int_0^t \int_{\mathbb{R}^{n-1}}\frac{e^{-\left|z\right|^2/4}}{(t-s)^{\frac{3}{2}}(4\pi)^{\frac{n}{2}}}\left(\sqrt{t-s}Dg(s,x)\cdot z+O((t-s)\left|z\right|^2)\right)dzds,
    \end{align*}
where we have again used Taylor's Theorem and the growth condition \eqref{ugrowth} in the last line, giving $O((t-s)\left|z\right|^2)\le C_2 (t-s)\left|z\right|^2e^{C_2\sqrt{t-s}\left|z\right|}$ for some constant $C_2>0$. Continuing the computation gives  
\begin{align*}
    &\int_{0}^{t}\int_{\partial \Omega}(g(s,y)-g(s,x))\nabla_{\nu_y}p(t-s,x,y)dyds\\&=\int_0^t \frac{1}{(t-s)(4\pi)^{\frac{n}{2}}} \int_{\mathbb{R}^{n-1}}e^{-\left|z\right|^2/4}Dg(s,x)\cdot zdzds+\int_0^t \int_{\mathbb{R}^{n-1}}\frac{e^{-\left|z\right|^2/4}}{(t-s)^{\frac{3}{2}}(4\pi)^{\frac{n}{2}}}O((t-s)\left|z\right|^2)dzds.
\end{align*}
Notice that, despite the diverging of the integral $\int_0^t\frac{1}{(t-s)}ds$, the first term vanishes on account of the single-variable integral $\int_{-\infty}^{\infty}e^{-z^2/4}zdz=0$. It is clear that the second integral converges. 
\end{proof}

\section{Gaussian flux distribution}
\label{sec:gaussian_flux_results}

\subsection{Two-dimensional}
If the boundary heat flux is given by 
\begin{align*}
    q(t,x) = \text{e}^{-x_1^2} ,
\end{align*}
we can determine the variation in surface temperature from
\begin{align}
\label{eq:Gaussian_2D_g}
    g(t,x_1) & = \int_0^{t} \int_{-\infty}^\infty  \frac{\text{e}^{-y_1^2}}{2 \pi s} \textrm{e}^{-(x_1-y_1)^2/4 s} \, dy_1 ds \nonumber \\
    &=\frac{1}{2\pi }\int_0^t \frac{1}{s}\int_{-\infty}^{\infty}\text{e}^{-y_1^2-(x_1-y_1)^2/4s} dy_1ds \nonumber \\
    &= \frac{\text{e}^{-x_1^2}}{\sqrt{\pi}}  \int_0^t \frac{1}{\sqrt{4 s^2 + s}}\text{exp} \left(\frac{4 x_1^2 s }{4s +  1} \right)  ds.
\end{align}

\subsection{Three-dimensional}
If the boundary heat flux is given by 
\begin{align*}
    q(t,x)= \text{e}^{-|x|^2} ,
\end{align*}
we can determine the variation in surface temperature from
\begin{align*}
    g(t,x)&=\int_0^t \int_{\partial \Omega} \frac{2 \text{e}^{-y^2} }{(4\pi s)^{3/2}}\text{e}^{-\left|x-y\right|^2/4s}dyds \\
    &= \frac{2}{(4\pi)^{3/2}} \int_0^t \frac{1 }{s^{3/2}} \int_{\partial \Omega}  \text{e}^{-y^2-\left|x-y\right|^2/4s}dyds .
\end{align*}
Introducing the change of variables $z\sqrt{s}=y-x$, and if $z = (z_1,z_2)$ and $x = (x_1,x_2)$ we have $dy=s.dz$ so we write
\begin{align*}
    g(t,x) &= \frac{2}{(4\pi)^{3/2}} \int_0^t \frac{1}{\sqrt{s}} \int_{\partial \Omega} \text{e}^{-(z\sqrt{s}+x)^2-z^2/4} dz ds \\
    &= \frac{2}{(4\pi)^{3/2}} \int_0^t \frac{1}{\sqrt{s}} \int_{\partial \Omega} \text{e}^{-(z\sqrt{s}+x)^2-z^2/4} dz ds .
\end{align*}
The boundary integral can be calculated to give the temperature as 
\begin{align}
    g(t,x) &= \frac{\text{e}^{-(x_1^2+x_2^2)}}{\sqrt{\pi}}  \int_0^t  \frac{1}{\sqrt{s} (4s + 1)} \text{exp} \left( \frac{4(x_1^2+x_2^2)s}{4s + 1} \right) ds.
\end{align}

\end{document}